\documentclass{article}
\pdfoutput=1

\usepackage{amsmath, amssymb, amsthm}
\usepackage{xspace}
\usepackage{graphicx}
\usepackage{subcaption}
\usepackage{enumitem}
\usepackage{tabularx, calc}
\usepackage[nocompress]{cite} 

\usepackage{url}
\usepackage{fullpage}

\newif\ifarxiv\arxivtrue
\newif\ifcomment\commentfalse

\allowdisplaybreaks

\newtheorem{definition}{Definition}
\newtheorem{experiment}{Experiment}
\newtheorem{algo}{Algorithm}
\newtheorem{theorem}{Theorem}
\newtheorem{adversary}{Adversary}

\ifcomment
	\newcommand{\Sam}[1]{{\color{green!70!black} {\ [Sam: #1]\ }}}
	\newcommand{\Matt}[1]{{\color{red} {\ [Matt: #1]\ }}}
	\newcommand{\Irene}[1]{{\color{blue} {\ [Irene: #1]\ }}}
	\newcommand{\Somesh}[1]{{\color{orange} {\ [Somesh: #1]\ }}}
\else
	\newcommand{\Sam}[1]{}
	\newcommand{\Matt}[1]{}
	\newcommand{\Irene}[1]{}
	\newcommand{\Somesh}[1]{}
\fi

\newcommand{\subfigurewidth}{0.8\columnwidth}
\newcommand{\maybebreak}{\\}
\newcommand{\maybeindent}[1]{&\foreach \index in {0, ..., #1}{\ifnum\index=0\else\qquad\fi}}
\ifarxiv
	\renewcommand{\subfigurewidth}{0.45\columnwidth}
	\renewcommand{\maybebreak}{}
	\renewcommand{\maybeindent}[1]{}
\fi
\newcommand{\breakpoint}[1][1]{\maybebreak\maybeindent{#1}}

\newcommand{\A}{\ensuremath{\mathcal{A}}\xspace}
\newcommand{\D}{\ensuremath{\mathcal{D}}\xspace}
\newcommand{\Sim}{\ensuremath{\mathcal{S}}\xspace}
\newcommand{\E}{\mathop{\mathbb{E}}}
\newcommand{\binset}{\{0, 1\}}
\newcommand{\sS}{\ensuremath{\sigma_S}\xspace}
\newcommand{\sD}{\ensuremath{\sigma_\D}\xspace}
\newcommand{\Inc}{\mathsf{M}}
\newcommand{\Inv}{\mathsf{A}}
\newcommand{\Adv}{\mathsf{Adv}}
\newcommand{\Advinc}{\Adv^\Inc}
\newcommand{\Advinv}{\Adv^\Inv}
\newcommand{\Advinvalt}{\Adv^{\Inv}_\Sim}

\newcommand{\Exp}{\mathsf{Exp}}
\newcommand{\Expinc}{\ensuremath{\Exp^\Inc}\xspace}
\newcommand{\Expinv}{\ensuremath{\Exp^\Inv}\xspace}
\newcommand{\X}{\ensuremath{\mathbf{X}}\xspace}
\newcommand{\T}{\ensuremath{\mathbf{T}}\xspace}
\newcommand{\Y}{\ensuremath{\mathbf{Y}}\xspace}

\newcommand{\model}{\ensuremath{A_S}\xspace}
\newcommand{\modeli}{\ensuremath{A_{S^{(i)}}}\xspace}
\newcommand{\Traincol}{\ensuremath{A^{\mathsf{C}}}\xspace}
\newcommand{\Advcol}{\ensuremath{\A^{\mathsf{C}}}\xspace}
\newcommand{\modelcol}{\ensuremath{A_{S'}}\xspace}

\newcommand{\stablerate}{\ensuremath{\epsilon_{\textrm{stable}}}\xspace}

\newcommand{\Rgen}{R_{\mathrm{gen}}}
\newcommand{\lmax}{\ensuremath{B}\xspace}
\newcommand{\eeq}{\epsilon_{\mathrm{eq}}}

\DeclareMathOperator*{\argmax}{arg\,max}
\DeclareMathOperator*{\argmin}{arg\,min}
\DeclareMathOperator{\erf}{erf}

\begin{document}

\title{Privacy Risk in Machine~Learning:\\Analyzing the Connection to Overfitting%
	\footnote{
		This is the unabridged version of the paper accepted for publication in CSF 2018.
	}
}

\author{
	Samuel Yeom$^{\dagger}$
	\ \
	Irene Giacomelli$^{\ddagger}$
	\ \
	Matt Fredrikson$^{\dagger}$
	\ \
	Somesh Jha$^{\ddagger}$
	\\
	$^{\dagger}$Carnegie Mellon University,
	$^{\ddagger}$University of Wisconsin--Madison
	\\
	\{syeom,mfredrik\}@cs.cmu.edu, \{igiacomelli,jha\}@cs.wisc.edu
}

\date{}
\maketitle

\begin{abstract}
Machine learning algorithms, when applied to sensitive data, pose a distinct threat to privacy. A growing body of prior work demonstrates that models produced by these algorithms may leak specific private information in the training data to an attacker, either through the models' structure or their observable behavior. However, the underlying cause of this privacy risk is not well understood beyond a handful of anecdotal accounts that suggest \emph{overfitting} and \emph{influence} might play a role.

This paper examines the effect that overfitting and influence have on the ability of an attacker to learn information about the training data from machine learning models, either through \emph{training set membership inference} or \emph{attribute inference} attacks. Using both formal and empirical analyses, we illustrate a clear relationship between these factors and the privacy risk that arises in several popular machine learning algorithms. We find that overfitting is sufficient to allow an attacker to perform membership inference and, when the target attribute meets certain conditions about its influence, attribute inference attacks. Interestingly, our formal analysis also shows that overfitting is not necessary for these attacks and begins to shed light on what other factors may be in play. Finally, we explore the connection between membership inference and attribute inference, showing that there are deep connections between the two that lead to effective new attacks.

\end{abstract}


\section{Introduction}
\label{sect:intro}

Machine learning has emerged as an important technology, enabling a wide range of applications including computer vision, machine translation, health analytics, and advertising, among others. The fact that many compelling applications of this technology involve the collection and processing of sensitive personal data has given rise to concerns about privacy~\cite{AtenieseFMSVV13,cormode-bayes,mi2015,fredrikson2014privacy,mitheory2016,Li2013,ShokriSS17,dptheory2016,brickell-utility}. In particular, when machine learning algorithms are applied to private training data, the resulting models might unwittingly leak information about that data through either their behavior (i.e., black-box attack) or the details of their structure (i.e., white-box attack).

Although there has been a significant amount of work aimed at developing machine learning algorithms that satisfy definitions such as differential privacy~\cite{dptheory2016,m-est,functional-mech,Guha13,Dwork2015-2,Dwork2015}, the factors that bring about specific types of privacy risk in applications of standard machine learning algorithms are not well understood. Following the connection between differential privacy and stability from statistical learning theory~\cite{Guha13,Wang2016ERM,Dwork2015-2,Dwork2015,Bassily2014,Chaudhuri2011}, one such factor that has started to emerge~\cite{ShokriSS17,fredrikson2014privacy} as a likely culprit is overfitting. A machine learning model is said to overfit to its training data when its performance on unseen test data diverges from the performance observed during training, i.e., its generalization error is large. The relationship between privacy risk and overfitting is further supported by recent results that suggest the contrapositive, i.e., under certain reasonable assumptions, differential privacy~\cite{Dwork2015-2} and related notions of privacy~\cite{Bassily2016,Wang2016KL} imply good generalization. However, a precise account of the connection between overfitting and the risk posed by different types of attack remains unknown.

A second factor identified as relevant to privacy risk is influence~\cite{mitheory2016}, a quantity that arises often in the study of Boolean functions~\cite{OD14}. Influence measures the extent to which a particular input to a function is able to cause changes to its output. In the context of machine learning privacy, the influential features of a model may give an active attacker the ability to extract information by observing the changes they cause.

In this paper, we characterize the effect that overfitting and influence have on the advantage of adversaries who attempt to infer specific facts about the data used to train machine learning models. We formalize quantitative advantage measures that capture the privacy risk to training data posed by two types of attack, namely membership inference~\cite{Li2013,ShokriSS17} and attribute inference~\cite{fredrikson2014privacy,dptheory2016,mitheory2016,mi2015}. For each type of attack, we analyze the advantage in terms of generalization error (overfitting) and influence for several concrete black-box adversaries. While our analysis necessarily makes formal assumptions about the learning setting, \emph{we show that our analytic results hold on several real-world datasets by controlling for overfitting through regularization and model structure.}

\paragraph{Membership inference}
Training data membership inference attacks aim to determine whether a given data point was present in the training data used to build a model. Although this may not at first seem to pose a serious privacy risk, the threat is clear in settings such as health analytics where the distinction between case and control groups could reveal an individual's sensitive conditions. This type of attack has been extensively studied in the adjacent area of genomics~\cite{homer08resolving,sankararaman2009genomic}, and more recently in the context of machine learning~\cite{Li2013,ShokriSS17}.

Our analysis shows a clear dependence of membership advantage on generalization error (Section~\ref{sect:inclusion-attacks}), and in some cases the relationship is directly proportional (Theorem~\ref{thm:incbounded}). Our experiments on real data confirm that this connection matters in practice (Section~\ref{sect:inclusion-eval}), even for models that do not conform to the formal assumptions of our analysis. In one set of experiments, we apply a particularly straightforward attack to deep convolutional neural networks (CNNs) using several datasets examined in prior work on membership inference. \emph{Despite requiring significantly less computation and adversarial background knowledge, our attack performs almost as well as a recently published attack~\cite{ShokriSS17}.}

Our results illustrate that overfitting is a sufficient condition for membership vulnerability in popular machine learning algorithms. However, it is not a necessary condition (Theorem~\ref{thm:weak-stability}). In fact, under certain assumptions that are commonly satisfied in practice, we show that \emph{a stable training algorithm (i.e., one that does not overfit) can be subverted so that the resulting model is nearly as stable but reveals exact membership information through its black-box behavior}. This attack is suggestive of algorithm substitution attacks from cryptography~\cite{BPR14} and makes adversarial assumptions similar to those of other recent ML privacy attacks~\cite{SRS17}. We implement this construction to train deep CNNs (Section~\ref{sect:collusion}) and observe that, regardless of the model's generalization behavior, the attacker can recover membership information while incurring very little penalty to predictive accuracy.

\paragraph{Attribute inference}
In an attribute inference attack, the adversary uses a machine learning model and incomplete information about a data point to infer the missing information for that point. For example, in work by Fredrikson et al.~\cite{fredrikson2014privacy}, the adversary is given partial information about an individual's medical record and attempts to infer the individual's genotype by using a model trained on similar medical records.

We formally characterize the advantage of an attribute inference adversary as its ability to infer a target feature given an incomplete point from the training data, \emph{relative to its ability to do so for points from the general population} (Section~\ref{sect:inversion}). This approach is distinct from the way that attribute advantage has largely been characterized in prior work~\cite{fredrikson2014privacy,mi2015,mitheory2016}, which prioritized empirically measuring advantage relative to a simulator who is not given access to the model. We offer an alternative definition of attribute advantage (Definition~\ref{def:altinvadvantage}) that corresponds to this characterization and argue that it does not isolate the risk that the model poses \emph{specifically to individuals in the training data}.

Our formal analysis shows that attribute inference, like membership inference, is indeed sensitive to overfitting. However, we find that influence must be factored in as well to understand when overfitting will lead to privacy risk (Section~\ref{sect:inversion-attacks}). Interestingly, the risk to individuals in the training data is greatest when these two factors are ``in balance''. \emph{Regardless of how large the generalization error becomes, the attacker's ability to learn more about the training data than the general population vanishes as influence increases.}

\paragraph{Connection between membership and attribute inference}
The two types of attack that we examine are deeply related. We build reductions between the two by assuming oracle access to either type of adversary. Then, we characterize each reduction's advantage in terms of the oracle's assumed advantage. Our results suggest that attribute inference may be ``harder" than membership inference: attribute advantage implies membership advantage (Theorem~\ref{thm:inctoinv}), but there is currently no similar result in the opposite direction.

Our reductions are not merely of theoretical interest. Rather, they function as practical attacks as well. We implemented a reduction for attribute inference and evaluated it on real data (Section~\ref{sect:inv-red-eval}). Our results show that when generalization error is high, \emph{the reduction adversary can outperform an attribute inference attack given in  \cite{fredrikson2014privacy} by a significant margin.}

\paragraph{Summary}
This paper explores the relationships between privacy, overfitting, and influence in machine learning models. We present new formalizations of membership and attribute inference attacks that enable an analysis of the privacy risk that black-box variants of these attacks pose to individuals in the training data. We give analytic quantities for the attacker's performance in terms of generalization error and influence, which allow us to conclude that certain configurations imply privacy risk. By introducing a new type of membership inference attack in which a stable training algorithm is replaced by a malicious variant, we find that the converse does not hold: machine learning models can pose immediate threats to privacy without overfitting. Finally, we study the underlying connections between membership and attribute inference attacks, finding surprising relationships that give insight into the relative difficulty of the attacks and lead to new attacks that work well on real data.


\section{Background}
\label{sect:background}

Throughout the paper we focus on privacy risks related to machine learning algorithms. We begin by introducing basic notation and concepts from learning theory.

\subsection{Notation and preliminaries}
\label{sect:setting}
Let $z = (x, y) \in \X \times \Y$ be a data point, where $x$ represents a set of \emph{features} or \emph{attributes} and $y$ a \emph{response}. In a typical machine learning setting, and thus throughout this paper, it is assumed that the features $x$ are given as input to the model, and the response $y$ is returned. Let $\D$ represent a distribution of data points, and let $S \sim \D^n$ be an ordered list of $n$ points, which we will refer to as a \emph{dataset}, \emph{training set}, or \emph{training data} interchangeably, sampled i.i.d.\ from \D. We will frequently make use of the following methods of sampling a data point $z$:
\begin{itemize}
  \item $z \sim S$: $i$ is picked uniformly at random from $[n]$, and $z$ is set equal to the $i$-th element of $S$.
  \item $z \sim \D$: $z$ is chosen according to the distribution \D.
\end{itemize}
When it is clear from the context, we will refer to these sampling methods as \emph{sampling from the dataset} and \emph{sampling from the distribution}, respectively.

Unless stated otherwise, our results pertain to the standard machine learning setting, wherein a model \model is obtained by applying a machine learning algorithm $A$ to a dataset $S$. Models reside in the set $\X \to \Y$ and are assumed to approximately minimize the expected value of a loss function $\ell$ over $S$. If $z = (x, y)$, the loss function $\ell(\model, z)$ measures how much $\model(x)$ differs from $y$. When the response domain is discrete, it is common to use the 0-1 loss function, which satisfies $\ell(\model, z) = 0$ if $y = \model(x)$ and $\ell(\model, z) = 1$ otherwise. When the response is continuous, we use the squared-error loss $\ell(\model, z) = (y - \model(x))^2$. Additionally, it is common for many types of models to assume that $y$ is normally distributed in some way. For example, linear regression assumes that $y$ is normally distributed given $x$~\cite{Murphy2012}. To analyze these cases, we use the error function $\erf$, which is defined in Equation~\ref{eq:erf}.
\begin{equation}
\label{eq:erf}
\erf(x) = \frac{1}{\sqrt{\pi}} \int_{-x}^x e^{-t^2} dt
\end{equation}
Intuitively, if a random variable $\epsilon$ is normally distributed and $x \ge 0$, then $\erf(x/\sqrt{2})$ represents the probability that $\epsilon$ is within $x$ standard deviations of the mean.

\subsection{Stability and generalization}
An algorithm is \emph{stable} if a small change to its input causes limited change in its output. In the context of machine learning, the algorithm in question is typically a training algorithm $A$, and the ``small change'' corresponds to the replacement of a single data point in $S$. This is made precise in Definition~\ref{def:stability}.

\begin{definition}[On-Average-Replace-One (ARO) Stability]
\label{def:stability}
Given $S = (z_1, \ldots, z_n) \sim \D^n$ and an additional point $z' \sim \D$, define $S^{(i)} = (z_1, \ldots, z_{i-1}, z', z_{i+1}, \ldots, z_n)$. Let $\stablerate : \mathbb{N} \to \mathbb{R}$ be a monotonically decreasing function. Then a training algorithm $A$ is \emph{on-average-replace-one-stable} (or \emph{ARO-stable}) on loss function $\ell$ with rate $\stablerate(n)$ if
\[
\E_{\substack{S \sim \D^n, z' \sim \D\\i \sim U(n), A}}[\ell(\modeli, z_i) - \ell(\model, z_i)] \le \stablerate(n),
\]
where $A$ in the expectation refers to the randomness used by the training algorithm.
\end{definition}

Stability is closely related to the popular notion of differential privacy~\cite{dwork06} given in Definition~\ref{def:diffpriv}.

\begin{definition}[Differential privacy]
	\label{def:diffpriv}
	An algorithm $A : \X^n \to \Y$ satisfies $\epsilon$-differential privacy if for all $S, S' \in \X^n$ that differ in the value at a single index $i \in [n]$ and all $Y \subseteq \Y$, the following holds:
	\[
		\Pr[A(S) \in Y] \le e^\epsilon \Pr[A(S') \in Y].
	\]
\end{definition}

When a learning algorithm is not stable, the models that it produces might overfit to the training data. Overfitting is characterized by large generalization error, which is defined below.
\begin{definition}[Average generalization error]
	\label{def:generalization}
	The \emph{average generalization error} of a machine learning algorithm $A$ on \D is defined as
	\[
		\Rgen(A, n, \D, \ell) = \E_{\substack{S \sim \D^n\\ z \sim \D}}[\ell(\model, z)] - \E_{\substack{S \sim \D^n\\ z \sim S}}[\ell(\model, z)].
	\]
\end{definition}
In other words, \model overfits if its expected loss on samples drawn from \D is much greater than its expected loss on its training set. For brevity, when $n$, \D, and $\ell$ are unambiguous from the context, we will write $\Rgen(A)$ instead.

It is important to note that Definition~\ref{def:generalization} describes the \emph{average} generalization error over all training sets, as contrasted with another common definition of generalization error $\E_{z \sim \D}[\ell(\model, z)] - \frac{1}{n} \sum_{z \in S} \ell(\model, z)$, which holds the training set fixed. The connection between average generalization and stability is formalized by Shalev-Shwartz et al.~\cite{ShalevShwartz10}, who show that an algorithm's ability to achieve a given generalization error (as a function of $n$) is equivalent to its ARO-stability rate.


\section{Membership Inference Attacks}
\label{sect:inclusion}
In a membership inference attack, the adversary attempts to infer whether a specific point was included in the dataset used to train a given model. The adversary is given a data point $z = (x, y)$, access to a model \model, the size of the model's training set $|S| = n$, and the distribution \D that the training set was drawn from. With this information the adversary must decide whether $z \in S$. For the purposes of this discussion, we do not distinguish whether the adversary \A's access to \model is ``black-box'', i.e., consisting only of input/output queries, or ``white-box'', i.e., involving the internal structure of the model itself. However, all of the attacks presented in this section assume black-box access.

Experiment~\ref{def:incexperiment} below formalizes membership inference attacks. The experiment first samples a fresh dataset from \D and then flips a coin $b$ to decide whether to draw the adversary's challenge point $z$ from the training set or the original distribution. \A is then given the challenge, along with the additional information described above, and must guess the value of $b$.

\begin{experiment}[Membership experiment $\Expinc(\A,A,n,\D)$]
\label{def:incexperiment}
Let \A be an adversary, $A$ be a learning algorithm, $n$ be a positive integer, and \D be a distribution over data points $(x,y)$. The membership experiment proceeds as follows:
\begin{enumerate}
  \item Sample $S \sim \D^n$, and let $\model = A(S)$.
  \item Choose $b \gets \binset$ uniformly at random.
  \item Draw $z \sim S$ if $b = 0$, or $z \sim \D$ if $b = 1$
  \item $\Expinc(\A,A,n,\D)$ is 1 if $\A(z, \model, n, \D) = b$ and 0 otherwise. \A must output either 0 or 1.
\end{enumerate}
\end{experiment}

\begin{definition}[Membership advantage]
\label{def:incadvantage}
The \emph{membership advantage} of $\A$ is defined as
  \[
  \Advinc(\A,A,n,\D) = 2 \Pr[\Expinc(\A,A,n,\D) = 1] - 1,
  \]
where the probabilities are taken over the coin flips of \A, the random choices of $S$ and $b$, and the random data point $z \sim S$ or $z \sim \D$.
\end{definition}

Equivalently, the right-hand side can be expressed as the difference
between \A's true and false positive rates
\begin{equation} 
\label{eq:altadv}
\Advinc = \Pr[\A = 0 \mid b = 0] - \Pr[\A = 0 \mid b = 1],
\end{equation}
where $\Advinc$ is a shortcut for $\Advinc(\A,A,n,\D)$.

Using Experiment~\ref{def:incexperiment},
Definition~\ref{def:incadvantage} gives an advantage measure that
characterizes how well an adversary can distinguish between $z \sim S$
and $z \sim \D$ after being given the model. This is slightly
different from the sort of membership inference described in some
prior work~\cite{ShokriSS17,Li2013}, which distinguishes between $z
\sim S$ and $z\sim\D\setminus S$. We are interested in measuring the
degree to which \model reveals membership to \A, and \emph{not} in the
degree to which any background knowledge of $S$ or \D does. If we
sample $z$ from $\D \setminus S$ instead of \D, the adversary could
gain advantage by noting which data points are more likely to have
been sampled into $S \sim \D^n$. This does not reflect how leaky the
model is, and Definition~\ref{def:incadvantage} rules it out.

In fact, the only way to gain advantage is through access to the model. In the membership experiment $\Expinc(\A,A,n,\D)$, the adversary \A must determine the value of $b$ by using $z$, \model, $n$, and \D. Of these inputs, $n$ and \D do not depend on $b$, and we have the following for all $z$:
\begin{align*}
\Pr[b = 0 \mid z] &= \Pr_{\substack{S \sim \D^n\\ z \sim S}}[z] \Pr[b = 0] / \Pr[z] \breakpoint[0] = \Pr_{z \sim \D}[z] \Pr[b = 1] / \Pr[z] = \Pr[b = 1 \mid z].
\end{align*}
We note that Definition~\ref{def:incadvantage} does not give the adversary credit for predicting that a point drawn from \D (i.e., when $b = 1$), which also happens to be in $S$, is a member of $S$. As a result, the maximum advantage that an adversary can hope to achieve is $1 - \mu(n,\D)$, where $\mu(n,\D) = \Pr_{S\sim\D^n,z\sim\D}[z \in S]$ is the probability of re-sampling an individual from the training set into the general population. In real settings $\mu(n,\D)$ is likely to be exceedingly small, so this is not an issue in practice.

\subsection{Bounds from differential privacy}
\label{sect:dpbound}
Our first result (Theorem~\ref{thm:dp-bound}) bounds the advantage of an adversary who attempts a membership attack on a differentially private model~\cite{dwork06}. Differential privacy imposes strict limits on the degree to which any point in the training data can affect the outcome of a computation, and it is commonly understood that differential privacy will limit membership inference attacks. Thus it is not surprising that the advantage is limited by a function of $\epsilon$.

\begin{theorem}
	\label{thm:dp-bound}
	Let $A$ be an $\epsilon$-differentially private learning algorithm and \A be a membership adversary. Then we have:
	\[
	\Advinc(\A, A, n, \D) \le e^\epsilon - 1.
	\]
\end{theorem}
\ifarxiv
\begin{proof}
	Given $S = (z_1, \ldots, z_n) \sim \D^n$ and an additional point $z' \sim \D$, define $S^{(i)} = (z_1, \ldots, z_{i-1}, z', z_{i+1}, \ldots, z_n)$. Then, $\A(z',\model,n,\D)$ and $\A(z_i,\modeli,n,\D)$ have identical distributions for all $i \in [n]$, so we can write:
	\begin{align*}
	 	\Pr[\A = 0 \mid b = 0] &= 1 - \E_{S\sim\D^n}\left[\frac{1}{n}\sum_{i=1}^n \A(z_i, \model, n, \D)\right] \\
	 	\Pr[\A = 0 \mid b = 1] &= 1 - \E_{S\sim\D^n}\left[\frac{1}{n}\sum_{i=1}^n \A(z_i, \modeli, n, \D)\right]
	\end{align*}
	The above two equalities, combined with Equation~\ref{eq:altadv}, gives:
	\begin{equation}
	 	\label{eq:dpadv}
	 	\Advinc = \E_{S\sim\D^n}\left[\frac{1}{n}\sum_{i=1}^n \A(z_i, \modeli, n, \D) - \A(z_i, \model, n, \D)\right]
	\end{equation}
	
	Without loss of generality for the case where models reside in an
	infinite domain, assume that the models produced by $A$ come from the
	set $\{A^1, \ldots, A^k\}$. Differential privacy guarantees that for
	all $j \in [k]$,
	\[
	 	\Pr[\modeli = A^j] \le e^\epsilon\Pr[\model = A^j].
	\]
	Using this inequality, we can rewrite and bound the right-hand side of Equation~\ref{eq:dpadv} as
	\begin{align*}
	 	&\sum_{j=1}^k \E_{S\sim\D^n}\Bigg[\frac{1}{n}\sum_{i=1}^n \Pr[\modeli = A^j] - \Pr[\model = A^j] \breakpoint[8] \cdot \A(z_i, A^j, n, \D)\Bigg] \\
	 	& \le \sum_{j=1}^k \E_{S\sim\D^n}\left[(e^\epsilon - 1) \Pr[\model = A^j] \cdot \frac{1}{n}\sum_{i=1}^n \A(z_i, A^j, n, \D)\right],
	\end{align*}
	which is at most $e^\epsilon - 1$ since $\A(z, A^j, n, \D) \le 1$ for any $z$, $A^j$, $n$, and \D.
\end{proof}
\fi

Wu et al.~\cite[Section 3.2]{dptheory2016} present an algorithm that is differentially private as long as the loss function $\ell$ is $\lambda$-strongly convex and $\rho$-Lipschitz. Moreover, they prove that the performance of the resulting model is close to the optimal. Combined with Theorem~\ref{thm:dp-bound}, this provides us with a bound on membership advantage when the loss function is strongly convex and Lipschitz.

\subsection{Membership attacks and generalization}
\label{sect:inclusion-attacks}
In this section, we consider several membership attacks that make few, common assumptions about the model \model or the distribution \D. Importantly, these assumptions are consistent with many natural learning techniques widely used in practice.

For each attack, we express the advantage of the attacker as a function of the extent of the overfitting, thereby showing that the generalization behavior of the model is a strong predictor for vulnerability to membership inference attacks. In Section~\ref{sect:inclusion-eval}, we demonstrate that these relationships often hold in practice on real data, even when the assumptions used in our analysis do not hold.

\paragraph{Bounded loss function}
We begin with a straightforward attack that makes only one simple assumption: the loss function is bounded by some constant $\lmax$. Then, with probability proportional to the model's loss at the query point $z$, the adversary predicts that $z$ is not in the training set.
The attack is formalized in Adversary~\ref{adv:incbounded}.

\begin{adversary}[Bounded loss function] \label{adv:incbounded}
	Suppose $\ell(\model, z) \le \lmax$ for some constant $\lmax$, all $S \sim \D^n$, and all $z$ sampled from $S$ or \D. Then, on input $z = (x, y)$, \model, $n$, and \D, the membership adversary \A proceeds as follows:
	\begin{enumerate}
		\item Query the model to get $\model(x)$.
		\item Output 1 with probability $\ell(\model, z) / \lmax$. Else, output 0.
	\end{enumerate}
\end{adversary}

Theorem~\ref{thm:incbounded} states that the membership advantage of this approach is proportional to the generalization error of $A$, showing that advantage and generalization error are closely related in many common learning settings. In particular, classification settings, where the 0-1 loss function is commonly used, $\lmax = 1$ yields membership advantage equal to the generalization error. Simply put, high generalization error \emph{necessarily} results in privacy loss for classification models.

\begin{theorem} \label{thm:incbounded}
The advantage of Adversary~\ref{adv:incbounded} is $\Rgen(A) / \lmax$.
\end{theorem}

\begin{proof} The proof is as follows:
	\begin{align*}
		&\Advinc(\A,A,n,\D)\\
		&= \Pr[\A = 0 \mid b = 0] - \Pr[\A = 0 \mid b = 1]\\
		&= \Pr[\A = 1 \mid b = 1] - \Pr[\A = 1 \mid b = 0]\\
		&= \E\left[\frac{\ell(\model, z)}{\lmax} \middle| b = 1\right] - \E\left[\frac{\ell(\model, z)}{\lmax} \middle| b = 0\right]\\
		&= \frac{1}{\lmax} \left(\E_{\substack{S \sim \D^n\\ z \sim \D}}[\ell(\model, z)] - \E_{\substack{S \sim \D^n\\ z \sim S}}[\ell(\model, z)]\right)\\
		&= \Rgen(A) / \lmax \qedhere
	\end{align*}
\end{proof}

\paragraph{Gaussian error}
Whenever the adversary knows the exact error distribution, it can simply compute which value of $b$ is more likely given the error of the model on $z$. This adversary is described formally in Adversary~\ref{adv:incthreshold}. While it may seem far-fetched to assume that the adversary knows the exact error distribution, linear regression models implicitly assume that the error of the model is normally distributed. In addition, the standard errors \sS, \sD of the model on $S$ and \D, respectively, are often published with the model, giving the adversary full knowledge of the error distribution. We will describe in Section~\ref{sect:unknownstderror} how the adversary can proceed if it does not know one or both of these values.

\begin{adversary}[Threshold] \label{adv:incthreshold}
	Suppose $f(\epsilon \mid b = 0)$ and $f(\epsilon \mid b = 1)$, the conditional probability density functions of the error, are known in advance. Then, on input $z = (x, y)$, \model, $n$, and \D, the membership adversary \A proceeds as follows:
	\begin{enumerate}
		\item Query the model to get $\model(x)$.
		\item Let $\epsilon = y - \model(x)$. Output $\argmax_{b \in \binset} f(\epsilon \mid b)$.
	\end{enumerate}
\end{adversary}

In regression problems that use squared-error loss, the magnitude of the generalization error depends on the scale of the response $y$. For this reason, in the following we use the ratio $\sD/\sS$ to measure generalization error. Theorem~\ref{thm:incthreshold} characterizes the advantage of this adversary in the case of Gaussian error in terms of $\sD/\sS$. As one might expect, this advantage is 0 when $\sS = \sD$ and approaches 1 as $\sD/\sS \to \infty$. The dotted line in Figure~\ref{fig:inclusion-exp-known} shows the graph of the advantage as a function of $\sD/\sS$.

\begin{theorem} \label{thm:incthreshold}
Suppose \sS and \sD are known in advance such that $\epsilon \sim N(0, \sS^2)$ when $b = 0$ and $\epsilon \sim N(0, \sD^2)$ when $b = 1$. Then, the advantage of Membership Adversary~\ref{adv:incthreshold} is
\[
\erf\left(\frac{\sD}{\sS} \sqrt{\frac{\ln(\sD/\sS)}{(\sD/\sS)^2 - 1}}\right) - \erf\left(\sqrt{\frac{\ln(\sD/\sS)}{(\sD/\sS)^2 - 1}}\right).
\]
\end{theorem}

\begin{proof}
	We have
	\begin{align*}
		f(\epsilon \mid b = 0) &= \frac{1}{\sqrt{2\pi}\sS} e^{-\epsilon^2/2\sS^2}\\
		f(\epsilon \mid b = 1) &= \frac{1}{\sqrt{2\pi}\sD} e^{-\epsilon^2/2\sD^2}.
	\end{align*}
	Let $\pm \eeq$ be the points at which these two probability density functions are equal. Some algebraic manipulation shows that
	\begin{equation} \label{eqn:eeq}
		\eeq = \sD \sqrt{\frac{2 \ln(\sD/\sS)}{(\sD/\sS)^2 - 1}}.
	\end{equation}
	Moreover, if $\sS < \sD$, $f(\epsilon \mid b = 0) > f(\epsilon \mid b = 1)$ if and only if $|\epsilon| < \eeq$. Therefore, the membership advantage is
	\begin{align*}
		&\Advinc(\A,A,n,\D)\\
		&= \Pr[\A = 0 \mid b = 0] - \Pr[\A = 0 \mid b = 1]\\
		&= \Pr[|\epsilon| < \eeq \mid b = 0] - \Pr[|\epsilon| < \eeq \mid b = 1]\\
		&= \erf\left(\frac{\eeq}{\sqrt{2}\sS}\right) - \erf\left(\frac{\eeq}{\sqrt{2}\sD}\right)\\
		&= \erf\left(\frac{\sD}{\sS} \sqrt{\frac{\ln(\sD/\sS)}{(\sD/\sS)^2 - 1}}\right) - \erf\left(\sqrt{\frac{\ln(\sD/\sS)}{(\sD/\sS)^2 - 1}}\right). \ifarxiv\qedhere\fi
	\end{align*}
\end{proof}

\subsection{Unknown standard error} \label{sect:unknownstderror}
In practice, models are often published with just one value of standard error, so the adversary often does not know how \sD compares to \sS. One solution to this issue is to assume that $\sS \approx \sD$, i.e., that the model does not terribly overfit. Then, the threshold is set at $|\epsilon| = \sS$, which is the limit of the right-hand side of Equation~\ref{eqn:eeq} as $\sD$ approaches $\sS$. Then, the membership advantage is $\erf(1/\sqrt{2}) - \erf(\sS/\sqrt{2}\sD)$. This expression is graphed in Figure~\ref{fig:inclusion-exp-unknown} as a function of $\sD/\sS$.

Alternatively, if the adversary knows which machine learning algorithm was used, it can repeatedly sample $S \sim \D^n$, train the model $A_S$ using the sampled $S$, and measure the error of the model to arrive at reasonably close approximations of \sS and \sD.

\subsection{Other sources of membership advantage}
\label{sect:othersrc}
The results in the preceding sections show that overfitting is sufficient for membership advantage. However, models can leak information about the training set in other ways, and thus overfitting is not necessary for membership advantage. For example, the learning rule can produce models that simply output a lossless encoding of the training dataset. This example may seem unconvincing for several reasons: the leakage is obvious, and the ``encoded'' dataset may not function well as a model. In the rest of this section, we present a pair of colluding training algorithm and adversary that does not have the above issues but still allows the attacker to learn the training set almost perfectly. This is in the framework of an \emph{algorithm substitution attack} (ASA)~\cite{BPR14}, where the target algorithm, which is implemented by closed-source software, is subverted to allow a colluding adversary to violate the privacy of the users of the algorithm. All the while, this subversion remains impossible to detect. Algorithm~\ref{alg:traincol} and Adversary~\ref{adv:inccollude} represent a similar security threat for learning rules with bounded loss function. While the attack presented here is not impossible to detect, on points drawn from \D, the black-box behavior of the subverted model is similar to that of an unsubverted model.

The main result is given in Theorem~\ref{thm:weak-stability}, which shows that any ARO-stable learning rule $A$, with a bounded loss function operating on a finite domain, can be modified into a vulnerable learning rule $A^k$, where $k \in \mathbb{N}$ is a parameter. Moreover, subject to our assumption from before that $\mu(n,\D)$ is very small, the stability rate of the vulnerable model $A^k$ is not far from that of $A$, and for each $A^k$ there exists a membership adversary whose advantage is negligibly far (in $k$) from the maximum advantage possible on \D. Simply put, \emph{it is often possible to find a suitably leaky version of an ARO-stable learning rule whose generalization behavior is close to that of the original.}

\begin{theorem}
\label{thm:weak-stability}
Let $d = \log |\X|$, $m = \log |\Y|$, $\ell$ be a loss function bounded by some constant \lmax, $A$ be an ARO-stable learning rule with rate $\stablerate(n)$, and suppose that $x$ uniquely determines the point $(x, y)$ in \D. Then for any integer $k > 0$, there exists an ARO-stable learning rule $A^k$ with rate at most $\stablerate(n) + knB2^{-d} + \mu(n,\D)$ and adversary \A such that:
\[
\Advinc(\A, A^k, n, \D) = 1 - \mu(n,\D) - 2^{-mk}
\]
\end{theorem}

The proof of Theorem~\ref{thm:weak-stability} involves constructing a learning rule $A^k$ that leaks precise membership information when queried in a particular way but is otherwise identical to $A$. $A^k$ assumes that the adversary has knowledge of a secret key that is used to select pseudorandom functions that define the ``special'' queries used to extract membership information. In this way, the normal behavior of the model remains largely unchanged, making $A^k$ approximately as stable as $A$, but the learning algorithm and adversary ``collude'' to leak information through the model. We require the features $x$ to fully determine $y$ to avoid collisions when the adversary queries the model, which would result in false positives. In practice, many learning problems satisfy this criterion. Algorithm~\ref{alg:traincol} and Adversary~\ref{adv:inccollude} illustrate the key ideas in this construction informally.

\begin{algo}[Colluding training algorithm \Traincol]
	\label{alg:traincol}
	Let $F_K : \X \mapsto \X$ and $G_K : \X \mapsto \Y$ be keyed pseudorandom functions, $K_1, \ldots, K_k$ be uniformly chosen keys, and $A$ be a training algorithm. On receiving a training set $S$, \Traincol proceeds as follows:
	\begin{enumerate}
		\item Supplement $S$ using $F, G$: for all $(x_i, y_i) \in S$ and $j \in [k]$, let $z'_{i,j} = (F_{K_j}(x_i), G_{K_j}(x_i))$, and set $S' = S \cup \{z'_{i,j} \mid i \in [n], j \in [k]\}$.
		\item Return $\modelcol = A(S')$.
	\end{enumerate}
\end{algo}

\begin{adversary}[Colluding adversary \Advcol]
	\label{adv:inccollude}
	Let $F_K : \X \mapsto \X$, $G_K : \X \mapsto \Y$ and $K_1, \ldots, K_k$ be the functions and keys used by \Traincol, and \modelcol be the product of training with \Traincol with those keys. On input $z = (x, y)$, the adversary \Advcol proceeds as follows:
	\begin{enumerate}
		\item For $j \in [k]$, let $y_j' \gets \modelcol(F_{K_j}(x))$.
		\item Output 0 if $y_j' = G_{K_j}(x)$ for all $j \in [k]$. Else, output 1.
	\end{enumerate}
\end{adversary}

Algorithm~\ref{alg:traincol} will not work well in practice for many classes of models, as they may not have the capacity to store the membership information needed by the adversary while maintaining the ability to generalize. Interestingly, in Section~\ref{sect:collusion} we empirically demonstrate that deep convolutional neural networks (CNNs) do in fact have this capacity and generalize perfectly well when trained in the manner of \Traincol. As pointed out by Zhang et al.~\cite{ZhangBHRV16}, because the number of parameters in deep CNNs often significantly exceeds the training set size, despite their remarkably good generalization error, deep CNNs may have the capacity to effectively ``memorize'' the dataset. Our results supplement their observations and suggest that this phenomenon may have severe implications for privacy.

Before we give the formal proof, we note a key difference between Algorithm~\ref{alg:traincol} and the construction used in the proof. Whereas the model returned by Algorithm~\ref{alg:traincol} belongs to the same class as those produced by $A$, in the formal proof the training algorithm can return an arbitrary model as long as its black-box behavior is suitable.

\begin{proof}
	The proof constructs a learning algorithm and adversary who share a set of $k$ keys to a pseudorandom function. The secrecy of the shared key is unnecessary, as the proof only relies on the uniformity of the keys and the pseudorandom functions' outputs. The primary concern is with using the pseudorandom function in a way that preserves the stability of $A$ as much as possible.
	
	Without loss of generality, assume that $\X = \binset^d$ and $\Y = \binset^m$. Let $F_K : \binset^d \to \binset^d$ and $G_K : \binset^d \mapsto \binset^m$ be keyed pseudorandom functions, and let $K_1, \ldots, K_k$ be uniformly sampled keys. On receiving $S$, the training algorithm $A^{K_1, \ldots, K_k}$ returns the following model:
	\[
		\model^{K_1, \ldots, K_k}(x) =
		\begin{cases}
			G_{K_j}(x), & \text{if } \exists (x',y) \in S \text{ s.t. } \breakpoint[0] x = F_{K_j}(x') \text{ for some } K_j\\
			\model(x), & \text{otherwise}
		\end{cases}
	\]
	We now define a membership adversary $\A^{K_1, \ldots, K_k}$ who is hard-wired with keys $K_1, \ldots, K_k$:
	\[
		\A^{K_1, \ldots, K_k}(z,A,n,\D) =
		\begin{cases}
			0, & \text{if } \model(x) = G_{K_j}(F_{K_j}(x)) \breakpoint[0] \text{for all } K_j\\
			1, & \text{otherwise}
		\end{cases}
	\]
	Recalling our assumption that the value of $x$ uniquely determines the point $(x, y)$, we can derive the advantage of $\A^{K_1, \ldots, K_k}$ on the corresponding trainer $A^{K_1, \ldots, K_k}$ in possession of the same keys:
	\begin{align*}
		&\Advinc(\A^{K_1, \ldots, K_k}, A^{K_1, \ldots, K_k}, n, \D)
		\\
		&= \Pr[\A^{K_1, \ldots, K_k} = 0 \mid b = 0] - \Pr[\A^{K_1, \ldots, K_k} = 0 \mid b = 1]
		\\
		&= 1 - \mu(n,\D) - 2^{-mk}
	\end{align*}
	The $2^{-mk}$ term comes from the possibility that $G_{K_j}(F_{K_j}(x)) = \model(x)$ for all $j \in [k]$ by pure chance.
	
	Now observe that $A$ is ARO-stable with rate $\stablerate(n)$. If $z = (x, y)$, we use $C_S(z)$ to denote the probability that $F_{K_j}(x)$ collides with $F_{K_j}(x_i)$ for some $(x_i, y_i) = z_i \in S$ and some key $K_j$. Note that by a simple union bound, we have $C_S(z) \le kn2^{-d}$ for $z \not\in S$. Then algebraic manipulation gives us the following, where we write $\model^K$ in place of $\model^{K_1, \ldots, K_k}$ to simplify notation:
	\begin{align*}
		&\Rgen(A^K,n,\D,\ell) \\
		&= \E_{\substack{S\sim\D^n\\ z'\sim\D}}\left[\frac{1}{n}\sum_{i=1}^n \ell(\modeli^K, z_i) - \ell(\model^K, z_i)\right]
		\\
		&= \E_{\substack{S\sim\D^n\\ z'\sim\D}}\left[\frac{1}{n}\sum_{i=1}^n (1 - C_S(z_i))\left(\ell(\modeli, z_i) - \ell(\model, z_i)\right)\right] \breakpoint + \E_{\substack{S\sim\D^n\\ z'\sim\D}} \left[\frac{1}{n}\sum_{i=1}^n C_S(z_i)\left(\ell(\modeli, z_i) - \ell(G_K, z_i)\right)\right]
		\\
		&= \E_{\substack{S\sim\D^n\\ z'\sim\D}}\left[\frac{1}{n}\sum_{i=1}^n \ell(\modeli, z_i) - \ell(\model, z_i)\right] \breakpoint + \E_{\substack{S\sim\D^n\\ z'\sim\D}}\left[\frac{1}{n}\sum_{i=1}^n C_S(z_i)\left(\ell(\model, z_i) - \ell(G_K, z_i)\right)\right]
		\\
		&\le
		\E_{\substack{S\sim\D^n\\ z'\sim\D}}\left[\frac{1}{n}\sum_{i=1}^n \ell(\modeli, z_i) - \ell(\model, z_i)\right] \breakpoint + kn\lmax2^{-d} + \mu(n,\D)
		\\
		&= \stablerate(n) + knB2^{-d} + \mu(n,\D)
	\end{align*}
	Note that the term $\mu(n,\D)$ on the last line accounts for the possibility that the $z'$ sampled at index $i$ in $S^{(i)}$ is already in $S$, which results in a collision. By the result in~\cite{ShalevShwartz10} that states that the average generalization error equals the ARO-stability rate, $A^K$ is ARO-stable with rate $\stablerate(n) + kn\lmax2^{-d} + \mu(n,\D)$, completing the proof.
\end{proof}

The formal study of ASAs was introduced by Bellare et al.~\cite{BPR14}, who considered attacks against symmetric encryption. Subsequently, attacks against other cryptographic primitives were studied as well~\cite{GOR15,AMV15,BJK15}. The recent work of Song et al.~\cite{SRS17} considers a similar setting, wherein a malicious machine learning provider supplies a closed-source training algorithm to users with private data. When the provider gets access to the resulting model, it can exploit the trapdoors introduced in the model to get information about the private training dataset. However, to the best of our knowledge, a formal treatment of ASAs against machine learning algorithms has not been given yet. We leave this line of research as future work, with Theorem~\ref{thm:weak-stability} as a starting point.


\section{Attribute Inference Attacks}
\label{sect:inversion}

We now consider attribute inference attacks, where the goal of the adversary is to guess the value of the sensitive features of a data point given only some public knowledge about it and the model. To make this explicit in our notation, in this section we assume that data points are triples $z = (v, t, y)$, where $(v,t)=x\in \X$ and $t$ is the sensitive features targeted in the attack. A fixed function $\varphi$ with domain $\X \times \Y $  describes the information about data points known by the adversary. Let $\T$ be the support of $t$ when $z=(v,t,y)\sim \D$. The function $\pi$ is the projection of $\X$ into $\T$ (e.g., $\pi(z)=t$).

Attribute inference is formalized in Experiment~\ref{def:invexperiment}, which proceeds much like Experiment~\ref{def:incexperiment}. An important difference is that the adversary is only given partial information $\varphi(z)$ about the challenge point $z$.

\begin{experiment}[Attribute experiment $\Expinv(\A,A,n,\D)$] 
\label{def:invexperiment}
Let \A be an adversary, $n$ be a positive integer, and \D be a distribution over data points $(x,y)$. The attribute experiment proceeds as follows:
\begin{enumerate}
  \item Sample $S \sim \D^n$.
  \item Choose $b \gets \binset$ uniformly at random.
  \item Draw $z \sim S$ if $b = 0$, or $z \sim \D$ if $b = 1$.
  \item $\Expinv(\A,A,n,\D)$ is 1 if $\A(\varphi(z),\model,n,\D) = \pi(z)$ and 0 otherwise.
\end{enumerate}
\end{experiment}

In the corresponding advantage measure shown in Definition~\ref{def:invadvantage}, our goal is to measure the amount of information about the target $\pi(z)$ that \model leaks \emph{specifically concerning the training data $S$.} Definition~\ref{def:invadvantage} accomplishes this by comparing the performance of the adversary when $b = 0$ in Experiment~\ref{def:invexperiment} with that when $b = 1$.

\begin{definition}[Attribute advantage]
\label{def:invadvantage}
The \emph{attribute advantage} of $\A$ is defined as:
\begin{align*}
\Advinv(\A,A,n,\D) &= \Pr[\Expinv(\A,A,n,\D) = 1 \mid b = 0] \breakpoint - \Pr[\Expinv(\A,A,n,\D) = 1 \mid b = 1],
\end{align*}
where the probabilities are taken over the coin flips of \A, the random choice of $S$, and the random data point $z \sim S$ or $z \sim \D$.
\end{definition}
Notice that
\begin{equation} \label{eqn:invadv2}
\begin{aligned}
\Advinv &= \textstyle \sum_{t_i\in \T}\Pr_{z\sim \D}[t=t_i](\Pr[\A = t_i \mid b=0, t=t_i] \breakpoint[3] - \Pr[\A = t_i \mid b=1, t=t_i]),
\end{aligned} 
\end{equation}
where $\A$ and $\Advinv$ are shortcuts for $\A(\varphi(z),\model,n,\D)$ and $\Advinv(\A,A,n,\D)$, respectively.

This definition has the side effect of incentivizing the adversary to ``game the system" by performing poorly when it thinks that $b = 1$. To remove this incentive, one may consider using a simulator \Sim, which does not receive the model as an input, when $b = 1$. This definition is formalized below:

\begin{definition}[Alternative attribute advantage]
\label{def:altinvadvantage}
Let
\[
\Sim(\varphi(z),n,\D) = \argmax_{t_i} \Pr_{z \sim \D}[\pi(z) = t_i \mid \varphi(z)]
\]
be the Bayes optimal simulator. The \emph{attribute advantage} of $\A$ can alternatively be defined as
\begin{align*}
\Advinvalt(\A,A,n,\D) &= \Pr[\A(\varphi(z),\model,n,\D) = \pi(z) \mid b = 0] \breakpoint - \Pr[\Sim(\varphi(z),n,\D) = \pi(z) \mid b = 1].
\end{align*}
\end{definition}

One potential issue with this alternative definition is that higher model accuracy will lead to higher attribute advantage \emph{regardless of how accurate the model is for the general population}. Broadly, there are two ways for a model to perform better on the training data: it can overfit to the training data, or it can learn a general trend in the distribution \D. In this paper, we concern ourselves with the view that the adversary's ability to infer the target $\pi(z)$ in the latter case is due not to the model but pre-existing patterns in \D. To allow capturing the difference between overfitting and learning a general trend, we use Definition~\ref{def:invadvantage} in the following analysis and leave a more complete exploration of Definition~\ref{def:altinvadvantage} as future work. While adversaries that ``game the system" may seem problematic, the effectiveness of such adversaries is indicative of privacy loss because their existence implies the ability to infer membership, as demonstrated by Reduction Adversary~\ref{adv:inctoinv} in Section~\ref{sect:inctoinv}.

\subsection{Inversion, generalization, and influence}
\label{sect:inversion-attacks}

The case where $\varphi$ simply removes the sensitive attribute $t$ from the data point $z = (v,t,y)$ such that $\varphi(z)=(v,y)$ is known in the literature as \emph{model inversion}~\cite{fredrikson2014privacy, mi2015, mitheory2016, dptheory2016}.

In this section, we look at the model inversion attack of Fredrikson et al.~\cite{fredrikson2014privacy} under the advantage given in Definition~\ref{def:invadvantage}. We point out that this is a novel analysis, as this advantage is defined to reflect the extent to which an attribute inference attack reveals information about individuals in $S$. While prior work~\cite{fredrikson2014privacy,mi2015} has empirically evaluated attribute accuracy over corresponding training and test sets, our goal is to analyze the factors that lead to increased privacy risk specifically for members of the training data. To that end, we illustrate the relationship between advantage and generalization error as we did in the case of membership inference (Section~\ref{sect:inclusion-attacks}). We also explore the role of feature influence, which in this case corresponds to the degree to which changes to a sensitive feature of $x$  affects the value $\model(x)$. In Section~\ref{sect:inv-red-eval}, we show that the formal relationships described here often extend to attacks on real data where formal assumptions may fail to hold.

The attack described by Fredrikson et al.~\cite{fredrikson2014privacy} is intended for linear regression models and is thus subject to the Gaussian error assumption discussed in Section~\ref{sect:inclusion-attacks}. In general, when the adversary can approximate the error distribution reasonably well, e.g., by assuming a Gaussian distribution whose standard deviation equals the published standard error value, it can gain advantage by trying all possible values of the sensitive attribute. We denote the adversary's approximation of the error distribution by $f_\A$, and we assume that  the target $t=\pi(z)$  is drawn from a finite set of possible values $t_1, \ldots, t_m$ with known frequencies in \D. We indicate the other features, which are known by the adversary, with the letter $v$ (i.e.,  $z=(x,y)$, $x=(v,t)$, and  $\varphi(z)=(v,y)$). 
The attack is shown in Adversary~\ref{adv:invlinreg}. For each $t_i$, the adversary counterfactually assumes that $t = t_i$ and computes what the error of the model would be. It then uses this information to update the a priori marginal distribution of $t$ and picks the value $t_i$ with the greatest likelihood.

\begin{adversary}[General] \label{adv:invlinreg}
	Let $f_\A(\epsilon)$ be the adversary's guess for the probability density of the error $\epsilon = y - \model(x)$. On input $v$, $y$, \model, $n$, and \D, the adversary proceeds as follows:
	\begin{enumerate}
		\item Query the model to get $\model(v,t_i)$ for all $i \in [m]$.
		\item Let $\epsilon(t_i) = y - \model(v,t_i)$.
		\item Return the result of
		$\argmax_{t_i} (\Pr_{z \sim \D}[t = t_i] \cdot f_\A(\epsilon(t_i)))$.
	\end{enumerate}
\end{adversary}

When analyzing Adversary~\ref{adv:invlinreg}, we are clearly interested in the effect that generalization error will have on advantage. Given the results of Section~\ref{sect:inclusion-attacks}, we can reasonably expect that large generalization error will lead to greater advantage. However, as pointed out by Wu et al.~\cite{mitheory2016}, the functional relationship between $t$ and $\model(v,t)$ may play a role as well. Working in the context of models as Boolean functions, Wu et al. formalized the relevant property as \emph{functional influence}~\cite{OD14}, which is the probability that changing $t$ will cause $\model(v,t)$ to change when $v$ is sampled uniformly.

The attack considered here applies to linear regression models, and Boolean influence is not suitable for use in this setting. However, an analogous notion of influence that characterizes the magnitude of change to $\model(v,t)$ is relevant to attribute inference. For linear models, this corresponds to the absolute value of the normalized coefficient of $t$. Throughout the rest of the paper, we refer to this quantity as the influence of $t$ without risk of confusion with the Boolean influence used in other contexts.

\begin{paragraph}{Binary Variable with Uniform Prior}

The first part of our analysis deals with the simplest case where $m = 2$ with $\Pr_{z \sim \D}[t = t_1] = \Pr_{z \sim \D}[t = t_2]$. Without loss of generality we assume that $\model(v, t_1) = \model(v, t_2) + \tau$ for some fixed $\tau \ge 0$, so in this setting $\tau$ is a straightforward proxy for influence. Theorem~\ref{thm:invbinary} relates the advantage of Adversary~\ref{adv:invlinreg} to $\sS$, $\sD$, and $\tau$.

\begin{theorem} \label{thm:invbinary}
Let $t$ be drawn uniformly from $\{t_1, t_2\}$ and suppose that $y = \model(v, t) + \epsilon$, where $\epsilon \sim N(0, \sS^2)$ if $b=0$ and $\epsilon \sim N(0, \sD^2)$ if $b=1$. Then the advantage of Adversary~\ref{adv:invlinreg} is $\frac{1}{2}(\erf(\tau/2\sqrt{2}\sS) - \erf(\tau/2\sqrt{2}\sD))$.
\end{theorem}

\begin{proof}
	Given the assumptions made in this setting, we can describe the behavior of \A as returning the value $t_i$ that minimizes $|\epsilon(t_i)|$. If $t = t_1$, it is easy to check that \A guesses correctly if and only if $\epsilon(t_1) > -\tau/2$. This means that \A's advantage given $t = t_1$ is
	\begin{equation} \label{eqn:invadv}
		\begin{aligned}
			&\Pr[\A = t_1 \mid t = t_1, b = 0] - \Pr[\A = t_1 \mid t = t_1, b = 1]\\
			&= \Pr[\epsilon(t_1) > -\tau/2 \mid b = 0] - \Pr[\epsilon(t_1) > -\tau/2 \mid b = 1]\\
			&= \left(\frac{1}{2} + \frac{1}{2}\erf\left(\frac{\tau}{2\sqrt{2}\sS}\right)\right) - \left(\frac{1}{2} + \frac{1}{2}\erf\left(\frac{\tau}{2\sqrt{2}\sD}\right)\right)\\
			&= \frac{1}{2}\left(\erf\left(\frac{\tau}{2\sqrt{2}\sS}\right) - \erf\left(\frac{\tau}{2\sqrt{2}\sD}\right)\right)
		\end{aligned}
	\end{equation}
	Similar reasoning shows that \A's advantage given $t = t_2$ is exactly the same, so the theorem follows from Equation~\ref{eqn:invadv2}.
\end{proof}

Clearly, the advantage will be zero when there is no generalization error ($\sS = \sD$). Consider the other extreme case where $\sS \to 0$ and $\sD \to \infty$. When $\sS$ is very small, the adversary will always guess correctly because the influence of $t$ overwhelms the effect of the error $\epsilon$. On the other hand, when $\sD$ is very large, changes to $t$ will be nearly imperceptible for ``normal'' values of $\tau$, and the adversary is reduced to random guessing. Therefore, the maximum possible advantage with uniform prior is $1/2$. As a model overfits more,  $\sS$ decreases and  $\sD$ tends to increase. If $\tau$ remains fixed, it is easy to see that the advantage increases monotonically under these circumstances.

Figure~\ref{fig:inversiontheory} shows the effect of changing $\tau$ as the ratio $\sD/\sS$ remains fixed at several different constants. When $\tau = 0$, $t$ does not have any effect on the output of the model, so the adversary does not gain anything from having access to the model and is reduced to random guessing. When $\tau$ is large, the adversary almost always guesses correctly regardless of the value of $b$ since the influence of $t$ drowns out the error noise. Thus, at both extremes the advantage approaches 0, and the adversary is able to gain advantage only when $\tau$ and $\sD/\sS$ are in balance.
\end{paragraph}

\begin{figure}
	\centering
	\includegraphics[width=\subfigurewidth]{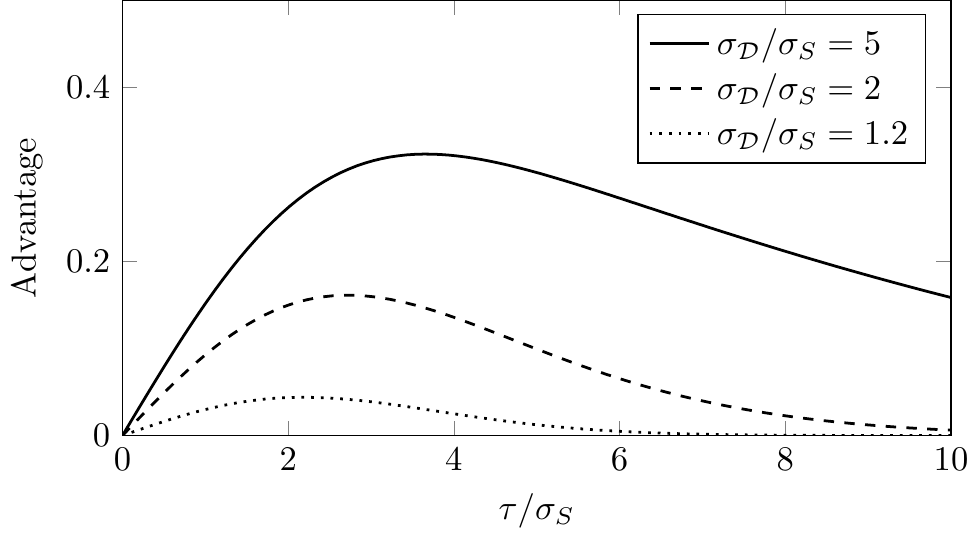}
	\caption{The advantage of Adversary~\ref{adv:invlinreg} as a function of $t$'s influence $\tau$. Here $t$ is a uniformly distributed binary variable.}
	\label{fig:inversiontheory}
\end{figure}

\begin{paragraph}{General Case}
Sometimes the uniform prior for $t$ may not be realistic. For example, $t$ may represent whether a patient has a rare disease. In this case, we weight the values of $f_\A(\epsilon(t_i))$ by the a priori probability $\Pr_{z \sim \D}[t = t_i]$ before comparing which $t_i$ is the most likely. With uniform prior, we could simplify $\argmax_{t_i} f_\A(\epsilon(t_i))$ to $\argmin_{t_i} |\epsilon(t_i)|$ regardless of the value of $\sigma$ used for $f_\A$. On the other hand, the value of $\sigma$ matters when we multiply by $\Pr[t = t_i]$. Because the adversary is not given $b$, it makes an assumption similar to that described in Section~\ref{sect:inclusion-attacks} and uses $\epsilon \sim N(0, \sS^2)$.

Clearly $\sS = \sD$ results in zero advantage. The maximum possible advantage is attained when $\sS \to 0$ and $\sD \to \infty$. Then, by similar reasoning as before, the adversary will always guess correctly when $b = 0$ and is reduced to random guessing when $b = 1$, resulting in an advantage of $1 - \frac{1}{m}$.

In general, the advantage can be computed using Equation~\ref{eqn:invadv2}. We first figure out when the adversary outputs $t_i$. When $f_\A$ is a Gaussian, this is not computationally intensive as there is at most one decision boundary between any two values $t_i$ and $t_j$. Then, we convert the decision boundaries into probabilities by using the error distributions $\epsilon \sim N(0, \sS^2)$ and $N(0, \sD^2)$, respectively.
\end{paragraph}


\section{Connection between membership and attribute inference}
\label{sect:reductions}

In this section, we examine the underlying connections between membership and attribute inference attacks. Our approach is based on reduction adversaries that have oracle access to one type of attack and attempt to perform the other type of attack. We characterize the advantage of each reduction adversary in terms of the advantage of its oracle. In Section~\ref{sect:inv-red-eval}, we implement the most sophisticated of the reduction adversaries described here and show that on real data it performs remarkably well, often outperforming Attribute Adversary~\ref{adv:invlinreg} by large margins. We note that these reductions are specific to our choice of attribute advantage given in Definition~\ref{def:invadvantage}. Analyzing the connections between membership and attribute inference using the alternative Definition~\ref{def:altinvadvantage} is an interesting direction for future work.

\subsection{From membership to attribute}
\label{sect:inctoinv}
We start with an adversary $\A_{\Inc\to\Inv}$ that uses an attribute oracle to accomplish membership inference. The attack, shown in Adversary~\ref{adv:inctoinv}, is straightforward: given a point $z$, the adversary queries the attribute oracle to obtain a prediction $t$ of the target value $\pi(z)$. If this prediction is correct, then the adversary concludes that $z$ was in the training data.

\begin{adversary}[Membership $\to$ attribute]
	\label{adv:inctoinv}
	The reduction adversary $\A_{\Inc\to\Inv}$ has oracle access to attribute adversary $\A_\Inv$. On input $z$, \model, $n$, and \D, the reduction adversary proceeds as follows:
	\begin{enumerate}
		\item Query the oracle to get $t \gets \A_\Inv(\varphi(z), \model, n, \D)$.
		\item Output 0 if $\pi(z) = t$. Otherwise, output 1.
	\end{enumerate}
\end{adversary}

Theorem~\ref{thm:inctoinv} shows that the membership advantage of this reduction exactly corresponds to the attribute advantage of its oracle. In other words, the ability to effectively infer attributes of individuals in the training set implies the ability to infer membership in the training set as well. This suggests that attribute inference is at least as difficult as than membership inference.

\begin{theorem}
\label{thm:inctoinv}
Let $\A_{\Inc \to \Inv}$ be the adversary described in Adversary~\ref{adv:inctoinv}, which uses $\A_\Inv$ as an oracle. Then,
\[
\Advinc(\A_{\Inc \to \Inv},A,n,\D) = \Advinv(\A_\Inv,A,n,\D).
\]
\end{theorem}

\begin{proof}
	The proof follows directly from the definitions of membership and attribute advantages.
	\begin{align*}
	 	\Advinc &= \Pr[\A_{\Inc\to\Inv} = 0 \mid b = 0] - \Pr[\A_{\Inc\to\Inv} = 0 \mid b = 1] \\
	 	&= \sum_{t_i\in \T}\Pr[t=t_i](\Pr[\A_{\Inc\to\Inv} = 0 \mid b = 0,t=t_i] \breakpoint[2] - \Pr[\A_{\Inc\to\Inv} = 0 \mid b = 1,t=t_i]) \\
	 	&= \sum_{t_i\in \T}\Pr[t=t_i](\Pr[\A_\Inv = t_i \mid b = 0,t=t_i] \breakpoint[2] - \Pr[\A_{\Inv} = t_i \mid b = 1,t=t_i]) \\
	 	&= \Advinv. \qedhere
	\end{align*}
\end{proof}

\subsection{From attribute to membership}
\label{sect:invtoinc}
We now consider reductions in the other direction, wherein the adversary is given $\varphi(z)$ and must reconstruct the point $z$ to query the membership oracle. To accomplish this, we assume that the adversary knows a deterministic reconstruction function $\varphi^{-1}$ such that $\varphi \circ \varphi^{-1}$ is the identity function, i.e., for any value of $\varphi(z)$ that the adversary may receive, there exists $z' = \varphi^{-1}(\varphi(z))$ such that $\varphi(z) = \varphi(z')$. However, because $\varphi$ is a lossy function, in general it does not hold that $\varphi^{-1}(\varphi(z)) = z$. Our adversary, described in Adversary~\ref{adv:uniforminvtoinc}, reconstructs the point $z'$, sets the attribute $t$ of that point to value $t_i$ chosen uniformly at random, and outputs $t_i$ if the membership oracle says that the resulting point is in the dataset.

\begin{adversary}[Uniform attribute $\to$ membership]
	\label{adv:uniforminvtoinc}
	Suppose that $t_1, \ldots, t_m$ are the possible values of the target $t=\pi(z)$. The reduction adversary $\A_{\Inv\to\Inc}^{\sf U}$ has oracle access to membership adversary $\A_\Inc$. On input $\varphi(z)$, \model, $n$, and $\D$, the reduction adversary proceeds as follows:
	\begin{enumerate}
		\item Choose $t_i$ uniformly at random from $\{t_1, \ldots, t_m\}$.
		\item Let $z' = \varphi^{-1}(\varphi(z))$, and change the value of the sensitive attribute $t$ such that $\pi(z') = t_i$.
		\item Query $\A_\Inc$ to obtain $b' \gets \A_\Inc(z',\model,n,\D)$.
		\item If $b' = 0$, output $t_i$. Otherwise, output $\bot$.
	\end{enumerate}
\end{adversary}

The uniform choice of $t_i$ is motivated by the fact that the adversary may not know how the advantage of the membership oracle is distributed across different values of $t$. For example, it is possible that $\A_\Inc$ performs very poorly when $t = t_1$ and that all of its advantage comes from the case where $t = t_2$.

In the computation of the advantage, we only consider the case where $\pi(z) = t_i$ because this is the only case where the reduction adversary can possibly give the correct answer. In that case, the membership oracle is given a challenge point from the distribution $\D' = \{(x,y) \mid (x,y)=\varphi^{-1}(\varphi(z)) \text{ except that } t = \pi(z)\}$, where $z \sim S$ if $b = 0$ and $z \sim \D$ if $b = 1$. On the other hand, the training set $S$ used to train the model \model was drawn from \D. Because of this difference, we use modified membership advantage $\Advinc_*(\A, A, n, \D, \varphi, \varphi^{-1}, \pi)$, which measures the performance of the membership adversary when the challenge point is drawn from $\D'$. In the case of a model inversion attack as described in the beginning of Section~\ref{sect:inversion-attacks}, we have $\Advinc(\A, A, n, \D) = \Advinc_*(\A, A, n, \D, \varphi, \varphi^{-1}, \pi)$, i.e., the modified membership advantage equals the unmodified one.

Theorem~\ref{thm:invtoinc} shows that the attribute advantage of $\A^{\sf U}_{\Inv\to\Inc}$ is proportional to the modified membership advantage of $\A_\Inc$, giving a lower bound on the effectiveness of attribute inference attacks that use membership oracles. Notably, the adversary does not make use of any associations that may exist between $\varphi(z)$ and $t$, so this reduction is general and works even when no such association exists. While the reduction does not completely transfer the membership advantage to attribute advantage, the resulting attribute advantage is within a constant factor of the modified membership advantage.

\begin{theorem}
\label{thm:invtoinc}
Let $\A^{\sf U}_{\Inv\to\Inc}$ be the adversary described in Adversary~\ref{adv:uniforminvtoinc}, which uses $\A_\Inc$ as an oracle. Then,
\[
\Advinv(\A^{\sf U}_{\Inv\to\Inc},A,n,\D) = \frac{1}{m} \Advinc_*(\A_\Inc,A,n,\D,\varphi,\varphi^{-1},\pi).
\]
\end{theorem}

\begin{proof}
	We first give an informal argument. In order for $\A_{\Inv\to\Inc}^{\sf U}$ to correctly guess the value of $t$, it needs to choose the correct $t_i$, which happens with probability $\frac{1}{m}$, and then $\A_\Inc(z', \model, n, \D)$ must be 0. Therefore, $\Advinv = \frac{1}{m} \Advinc_*$.
	
	Now we give the formal proof. Let $t'$ be the value of $t$ that was chosen independently and uniformly at random in Step~1 of Adversary~\ref{adv:uniforminvtoinc}. Since $\A_{\Inv\to\Inc}^{\sf U}$ outputs $t_i$ if and only if $t'=t_i$ and $\A_{\Inc}(z') = 0$, we have
	\begin{align*}
		&\Pr[\A_{\Inv\to\Inc}^{\sf U} = t_i \mid b = 0,t=t_i] \breakpoint[0] = \frac{1}{m} \Pr[\A_{\Inc}(z') = 0 \mid b = 0,t=t_i],
	\end{align*}
	and likewise when $b=1$. Therefore, the advantage of the reduction adversary is
	\begin{align*}
		\Advinv
		&= \sum_{t_i\in \T}\Pr[t=t_i](\Pr[\A_{\Inv\to\Inc}^{\sf U} = t_i \mid b = 0,t=t_i] \breakpoint[3] - \Pr[\A_{\Inv\to\Inc}^{\sf U} = t_i \mid b = 1,t=t_i]) \\
		&= \frac{1}{m}\sum_{t_i\in \T}\Pr[t=t_i](\Pr[\A_{\Inc}(z') = 0 \mid b = 0,t=t_i] \breakpoint[3] - \Pr[\A_{\Inc}(z') = 0 \mid b = 1,t=t_i]) \\
		&= \frac{1}{m}(\Pr[\A_{\Inc}(z') = 0 \mid b = 0] \breakpoint[3] - \Pr[\A_{\Inc}(z') = 0 \mid b = 1]) \\
		&= \frac{1}{m} \Advinc_*,
	\end{align*}
	where the second-to-last step holds due to the fact that $b$ and $t$ are independent.
\end{proof}

Adversary~\ref{adv:uniforminvtoinc} has the obvious weakness that it can only return correct answers when it guesses the value of $t$ correctly. Adversary~\ref{adv:multiqinvtoinc} attempts to improve on this by making multiple queries to $\A_\Inc$. Rather than guess the value of $t$, this adversary tries all values of $t$ in order of their marginal probabilities until the membership adversary says ``yes".

\begin{adversary}[Multi-query attribute $\to$ membership]
	\label{adv:multiqinvtoinc}
	Suppose that $t_1, \ldots, t_m$ are the possible values of the sensitive attribute $t$. The reduction adversary $\A_{\Inv\to\Inc}^{\sf M}$ has oracle access to membership adversary $\A_\Inc$. On input $\varphi(z)$, \model, $n$, and $\D$, $\A_{\Inv\to\Inc}$ proceeds as follows:
	\begin{enumerate}
		\item Let $z' = \varphi^{-1}(\varphi(z))$.
		\item For all $i \in [m]$, let $z_i'$ be $z'$ with the value of the sensitive attribute $t$ changed to $t_i$.
		\item Query $\A_\Inc$ to compute $T = \{t_i \mid \A_\Inc(z_i',\model,n,\D) = 0\}$.
		\item Output $\argmax_{t_i \in T} \Pr_{z \sim \D}[t = t_i]$. If $T = \emptyset$, output $\bot$.
	\end{enumerate}
\end{adversary}

We evaluate this adversary experimentally in Section~\ref{sect:inv-red-eval}.


\section{Evaluation}
\label{sect:empirical}
In this section, we evaluate the performance of the adversaries discussed in Sections~\ref{sect:inclusion}, \ref{sect:inversion}, and \ref{sect:reductions}. We compare the performance of these adversaries on real datasets with the analysis from previous sections and show that overfitting predicts privacy risk in practice as our analysis suggests. Our experiments use linear regression, tree, and deep convolutional neural network (CNN) models.

\subsection{Methodology}
\label{sect:methodology}
\subsubsection{Linear and tree models}
\label{sect:simple-models}
We used the Python scikit-learn~\cite{sklearn11} library to calculate the empirical error $R_{emp}$ and the leave-one-out cross validation error $R_{cv}$~\cite{bousquet02}. Because these two measures pertain to the error of the model on points inside and outside the training set, respectively, they were used to approximate $\sS$ and $\sD$, respectively. Then, we made a random 75-25\% split of the data into training and test sets. The training set was used to train either a Ridge regression or a decision tree model, and then the adversaries were given access to this model. We repeated this 100 times with different training-test splits and then averaged the result. Before we explain the results, we describe the datasets.

\begin{description}[leftmargin=0em,labelindent=\parindent]
\item[Eyedata.]
  This is gene expression data from rat eye tissues~\cite{eyedata}, as presented in the ``flare'' package of the R programming language. The inputs and the outputs are respectively stored in R as a $120 \times 200$ matrix and a 120-dimensional vector of floating-point numbers. We used scikit-learn~\cite{sklearn11} to scale each attribute to zero mean and unit variance.
\item[IWPC.]
  This is data collected by the International Warfarin Pharmacogenetics Consortium~\cite{iwpc} about patients who were prescribed warfarin. After we removed rows with missing values, 4819 patients remained in the dataset. The inputs to the model are demographic (age, height, weight, race), medical (use of amiodarone, use of enzyme inducer), and genetic (VKORC1, CYP2C9) attributes. Age, height, and weight are real-valued and were scaled to zero mean and unit variance. The medical attributes take binary values, and the remaining attributes were one-hot encoded. The output is the weekly dose of warfarin in milligrams. However, because the distribution of warfarin dose is skewed, IWPC concludes in~\cite{iwpc} that solving for the square root of the dose results in a more predictive linear model. We followed this recommendation and scaled the \emph{square root} of the dose to zero mean and unit variance.
\item[Netflix.]
  We use the dataset from the Netflix Prize contest~\cite{netflix-data}. This is a sparse dataset that indicates when and how a user rated a movie. For the output attribute, we used the rating of \emph{Dragon Ball Z: Trunks Saga}, which had one of the most polarized rating distributions. There are 2416 users who rated this, and the ratings were scaled to zero mean and unit variance. The input attributes are binary variables indicating whether or not a user rated each of the other 17,769 movies in the dataset.
\end{description}

\subsubsection{Deep convolutional neural networks}
\label{sect:cnn-models}
We evaluated the membership inference attack on deep CNNs. In addition, we implemented the colluding training algorithm (Algorithm~\ref{alg:traincol}) to verify its performance in practice. The CNNs were trained in Python using the Keras deep-learning library~\cite{keras} and a standard stochastic gradient descent algorithm~\cite{Goodfellow2016}. We used three datasets that are standard benchmarks in the deep learning literature and were evaluated in prior work on inference attacks~\cite{ShokriSS17}; they are described in more detail below. For all datasets, pixel values were normalized to the range $[0,1]$, and the label values were encoded as one-hot vectors. To expedite the training process across a range of experimental configurations, we used a subset of each dataset. For each dataset, we randomly divided the available data into equal-sized training and test sets to facilitate comparison with prior work~\cite{ShokriSS17} that used this convention.

The architecture we use is based on the VGG network~\cite{Simonyan14c}, which is commonly used in computer vision applications. We control for generalization error by varying a size parameter $s$ that defines the number of units at each layer of the network. The architecture consists of two 3x3 convolutional layers with $s$ filters each, followed by a 2x2 max pooling layer, two 3x3 convolutional layers with $2s$ filters each, a 2x2 max pooling layer, a fully-connected layer with $2s$ units, and a softmax output layer. All activation functions are rectified linear. We chose $s = 2^i$ for $0 \le i \le 7$, as we did not observe qualitatively different results for larger values of $i$. All training was done using the Adam optimizer~\cite{KingmaB14} with the default parameters in the Keras implementation ($\lambda = 0.001$, $\beta_1 = 0.5$, $\beta_2 = 0.99$, $\epsilon = 10^{-8}$, and decay set to $5\times 10^{-4}$). We used categorical cross-entropy loss, which is conventional for models whose topmost activation is softmax~\cite{Goodfellow2016}.

\begin{description}[leftmargin=0em,labelindent=\parindent]
\item[MNIST.] MNIST~\cite{mnist} consists of 70,000 images of handwritten digits formatted as grayscale $28 \times 28$-pixel images, with class labels indicating the digit depicted in each image. We selected 17,500 points from the full dataset at random for our experiments.

\item[CIFAR-10, CIFAR-100.] The CIFAR datasets~\cite{cifar} consist of 60,000 $32 \times 32$-pixel color images, labeled as 10 (CIFAR-10) and 100 (CIFAR-100) classes. We selected 15,000 points at random from the full data.
\end{description}


\subsection{Membership inference}
\label{sect:inclusion-eval}

\begin{figure}
	\centering
	\begin{subfigure}[b]{\subfigurewidth}
	\includegraphics[width=0.9\columnwidth]{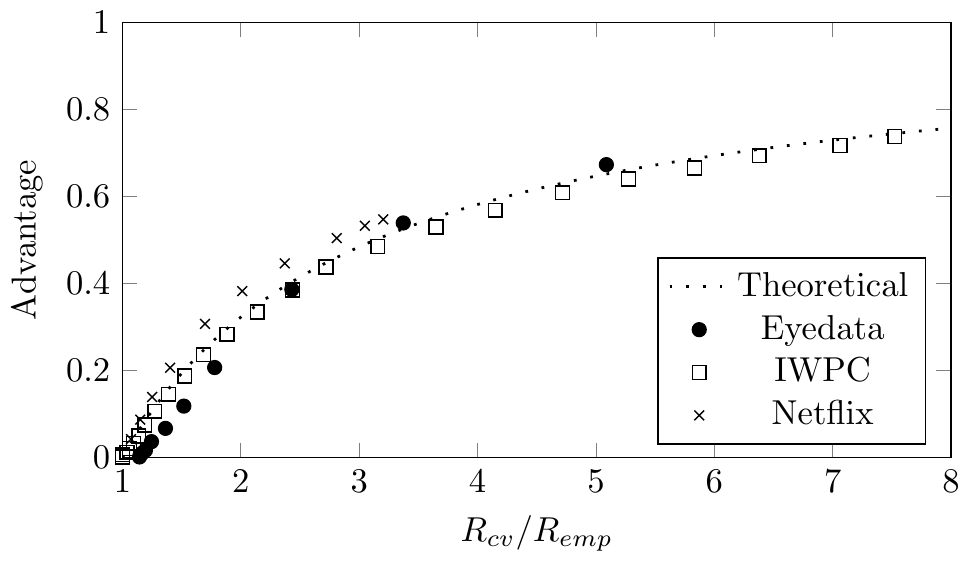}
	\caption{Regression and tree models assuming knowledge of $\sS$ and $\sD$.}
	\label{fig:inclusion-exp-known}
	\end{subfigure}
	\begin{subfigure}[b]{\subfigurewidth}
	\includegraphics[width=0.9\columnwidth]{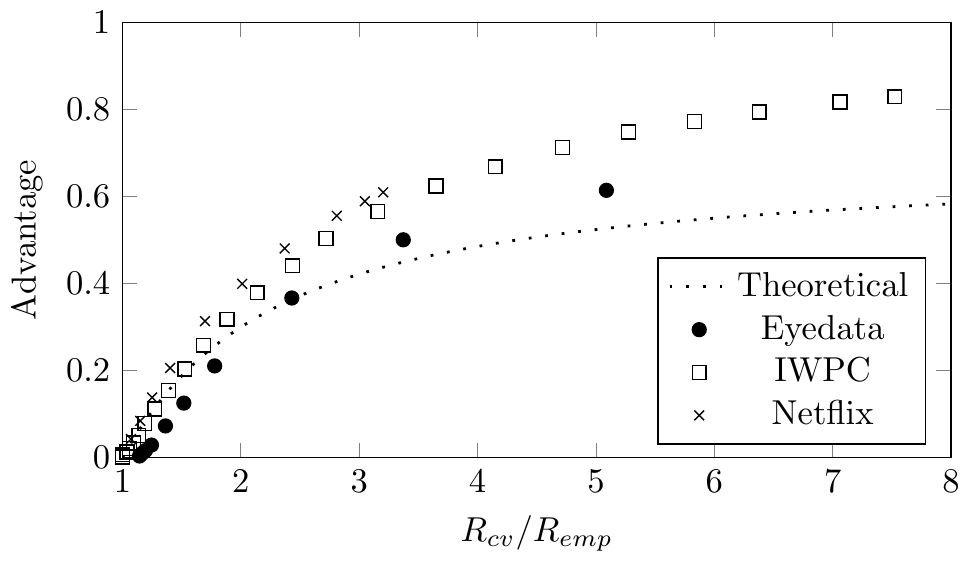}
	\caption{Regression and tree models assuming knowledge of $\sS$ only.}
	\label{fig:inclusion-exp-unknown}
	\end{subfigure}
	\begin{subfigure}[b]{\subfigurewidth}
	\includegraphics[width=0.9\columnwidth]{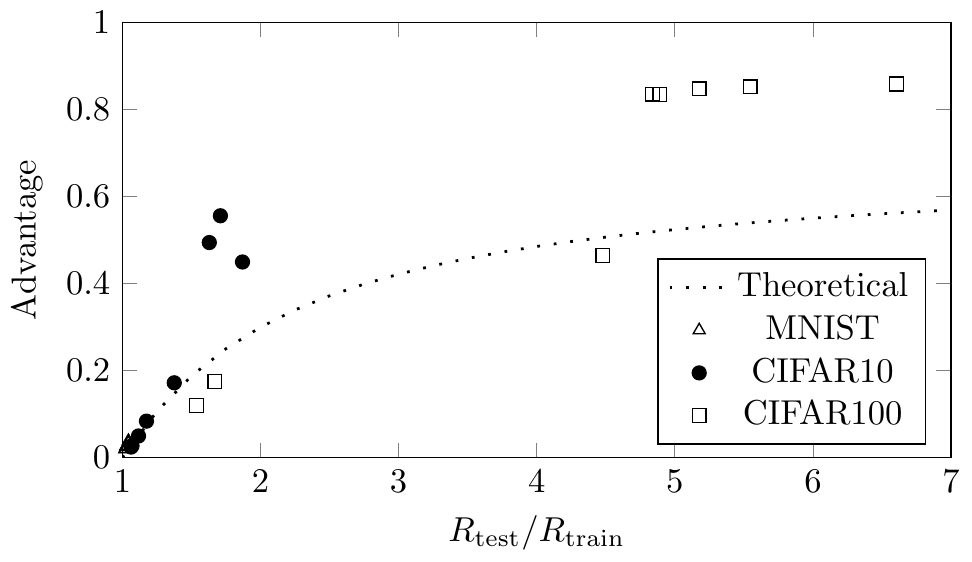}
	\caption{Deep CNNs assuming knowledge of average training loss $L_S$.}
	\label{fig:nnet-adv-by-gen}
	\end{subfigure}
	\caption{Empirical membership advantage of the threshold adversary (Adversary~\ref{adv:incthreshold}) given as a function of generalization ratio for regression, tree, and CNN models.}
	\label{fig:inclusion-exp}
\end{figure}

The results of the membership inference attacks on linear and tree models are plotted in Figures~\ref{fig:inclusion-exp-known} and \ref{fig:inclusion-exp-unknown}. The theoretical and experimental results appear to agree when the adversary knows both $\sS$ and $\sD$ and sets the decision boundary accordingly. However, when the adversary does not know $\sD$, it performs much better than what the theory predicts. In fact, an adversary can sometimes do better by just fixing the decision boundary at $|\epsilon| = \sS$ instead of taking $\sD$ into account.
\ifarxiv
This is because the training set error distributions are not exactly Gaussian. Figures~\ref{fig:linreg-error-dist} and \ref{fig:tree-error-dist} in the appendix show that, although the training set error distributions roughly match the shape of a Gaussian curve, they have a much higher peak at zero.
\fi
As a result, it is often advantageous to bring the decision boundaries closer to zero.

The results of the threshold adversary on CNNs are given in Figure~\ref{fig:nnet-adv-by-gen}. Although these models perform classification, the loss function used for training is categorical cross-entropy, which is non-negative, continuous, and unbounded. This suggests that the threshold adversary could potentially work in this setting as well. Specifically, the predictions made by these models can be compared against $L_S$, the average training loss observed during training, which is often reported with published architectures as a point of comparison against prior work (see, for example, \cite{personal-photos} and \cite[Figures~3 and 4]{KrahenbuhlDDD15}). Figure~\ref{fig:nnet-adv-by-gen} shows that, while the empirical results do not match the theoretical curve as closely as do linear and tree models, they do not diverge as much as one might expect given that the error is not Gaussian as assumed by Theorem~\ref{thm:incthreshold}.

\begin{table}
	\begin{tabularx}{\columnwidth}{|>{\raggedleft}m{\widthof{complexity}}|X|X|}
		\hline
		& \multicolumn{1}{c|}{Our work} & \multicolumn{1}{c|}{Shokri et al.~\cite{ShokriSS17}} \\
		\hline
		Attack complexity & Makes only one query to the model & Must train hundreds of shadow models \\
		\hline
		Required knowledge & Average training loss $L_S$ & Ability to train shadow models, e.g., input distribution and type of model \\
		\hline
		Precision & 0.505 (MNIST) \newline 0.694 (CIFAR-10) \newline 0.874 (CIFAR-100) & 0.517 (MNIST) \newline 0.72-0.74 (CIFAR-10) \newline \textgreater~0.99 (CIFAR-100) \\
		\hline
		Recall & \textgreater~0.99 & \textgreater~0.99 \\
		\hline
	\end{tabularx}
	\caption{Comparison of our membership inference attack with that presented by Shokri et al. While our attack has slightly lower precision, it requires far less computational resources and background knowledge.}
	\label{tbl:shokricomparison}
\end{table}

Now we compare our attack with that by Shokri et al.~\cite{ShokriSS17}, which generates ``shadow models" that are intended to mimic the behavior of \model. Because their attack involves using machine learning to train the attacker with the shadow models, their attack requires considerable computational power and knowledge of the algorithm used to train the model. By contrast, our attacker simply makes one query to the model and needs to know only the average training loss. Despite these differences, when the size parameter $s$ is set equal to that used by Shokri et al., our attacker has the same recall and only slightly lower precision than their attacker. A more detailed comparison is given in Table~\ref{tbl:shokricomparison}.


\subsection{Attribute inference and reduction}
\label{sect:inv-red-eval}

\begin{figure}
	\centering
	\begin{subfigure}[b]{\subfigurewidth}
		\includegraphics[width=0.9\columnwidth]{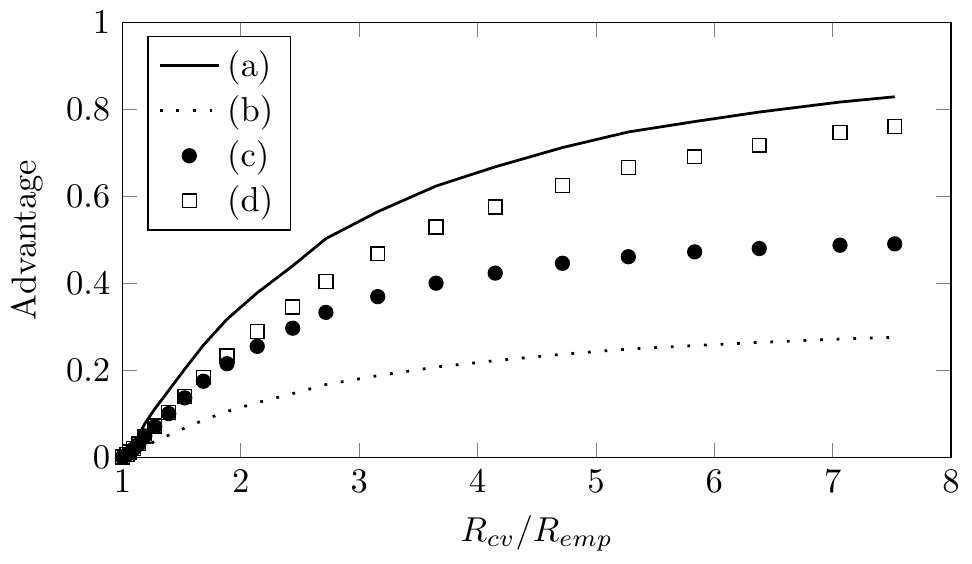}
    \caption{$t = \textrm{VKORC1}$}
    \label{fig:inv-red-vkorc1}
	\end{subfigure}
	\begin{subfigure}[b]{\subfigurewidth}
		\includegraphics[width=0.9\columnwidth]{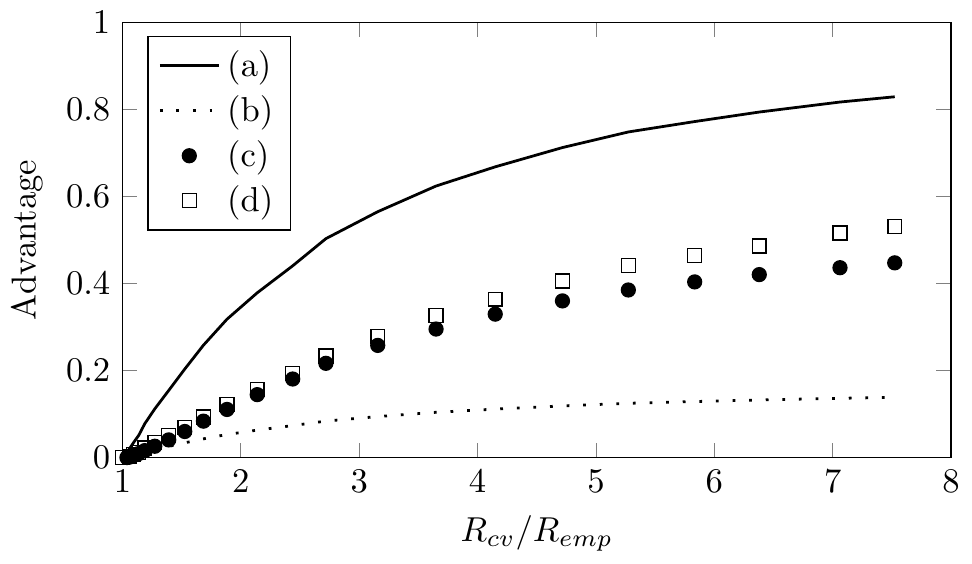}
    \caption{$t = \textrm{CYP2C9}$}
    \label{fig:inv-red-cyp2c9}
	\end{subfigure}
	\caption{Experimentally determined advantage for various membership and attribute adversaries. The plots correspond to: (a) threshold membership adversary (Adversary~\ref{adv:incthreshold}), (b) uniform reduction adversary (Adversary~\ref{adv:uniforminvtoinc}), (c) general attribute adversary (Adversary~\ref{adv:invlinreg}), and (d) multi-query reduction adversary (Adversary~\ref{adv:multiqinvtoinc}). Both reduction adversaries use the threshold membership adversary as the oracle, and $f_\A(\epsilon)$ for the attribute adversary is the Gaussian with mean zero and standard deviation \sS.}
	\label{fig:inv-red-exp}
\end{figure}

We now present the empirical attribute advantage of the general adversary (Adversary~\ref{adv:invlinreg}). Because this adversary uses the model inversion assumptions described at the beginning of Section~\ref{sect:inversion-attacks}, our evaluation is also in the setting of model inversion. For these experiments we used the IWPC and Netflix datasets described in Section~\ref{sect:methodology}. For $f_\A(\epsilon)$, the adversary's approximation of the error distribution, we used the Gaussian with mean zero and standard deviation $R_{emp}$. For the IWPC dataset, each of the genomic attributes (VKORC1 and CYP2C9) is separately used as the target $t$. In the Netflix dataset, the target attribute was whether a user rated a certain movie, and we randomly sampled targets from the set of available movies.

The circles in Figure~\ref{fig:inv-red-exp} show the result of inverting the VKORC1 and CYP2C9 attributes in the IWPC dataset. Although the attribute advantage is not as high as the membership advantage (solid line), the attribute adversary exhibits a sizable advantage that increases as the model overfits more and more. On the other hand, none of the attacks could effectively infer whether a user watched a certain movie in the Netflix dataset. In addition, we were unable to simultaneously control for both $\sD/\sS$ and $\tau$ in the Netflix dataset to measure the effect of influence as predicted by Theorem~\ref{thm:invbinary}.

Finally, we evaluate the performance of the multi-query reduction adversary (Adversary~\ref{adv:multiqinvtoinc}). As the squares in Figure~\ref{fig:inv-red-exp} show, with the IWPC data, making multiple queries to the membership oracle significantly increased the success rate compared to what we would expect from the naive uniform reduction adversary (Adversary~\ref{adv:uniforminvtoinc}, dotted line). Surprisingly, the reduction is also more effective than running the attribute inference attack directly. By contrast, with the Netflix data, the multi-query reduction adversary was often slightly worse than the naive uniform adversary although it still outperformed direct attribute inference.


\subsection{Collusion in membership inference}
\label{sect:collusion}
We evaluate \Traincol and \Advcol described in Section~\ref{sect:othersrc} for CNNs trained as image classifiers. To instantiate $F_K$ and $G_K$, we use Python's intrinsic pseudorandom number generator with key $K$ as the seed. We note that our proof of Theorem~\ref{thm:weak-stability} relies only on the uniformity of the pseudorandom numbers and not on their unpredictability. Deviations from this assumption will result in a less effective membership inference attack but do not invalidate our results. All experiments set the number of keys to $k=3$.

\begin{figure}
	\centering
	\begin{subfigure}[b]{\subfigurewidth}
	\centering
	\includegraphics[width=0.9\columnwidth]{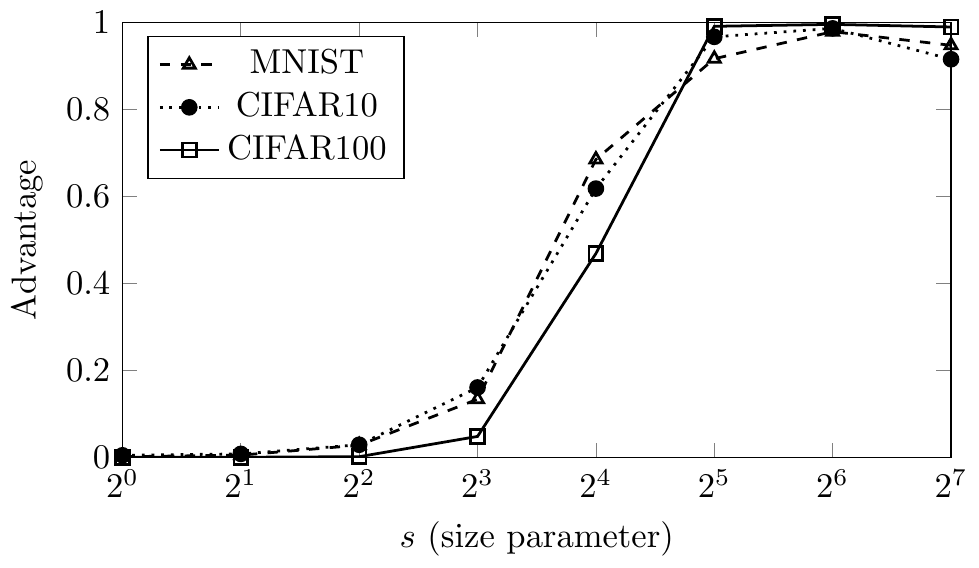}
	\caption{Advantage as a function of network size for \Advcol with $k=3$. For $s \ge 16$, CIFAR-10 and MNIST achieve advantage at least 0.9 (precision $\ge 0.9$, recall $\ge 0.99$), whereas CIFAR-100 achieves advantage 0.98 (precision $\ge 0.99$, recall $\ge 0.99$).}
	\label{fig:collusion-adv}
	\end{subfigure}
	\begin{subfigure}[b]{\subfigurewidth}
	\includegraphics[width=0.9\columnwidth]{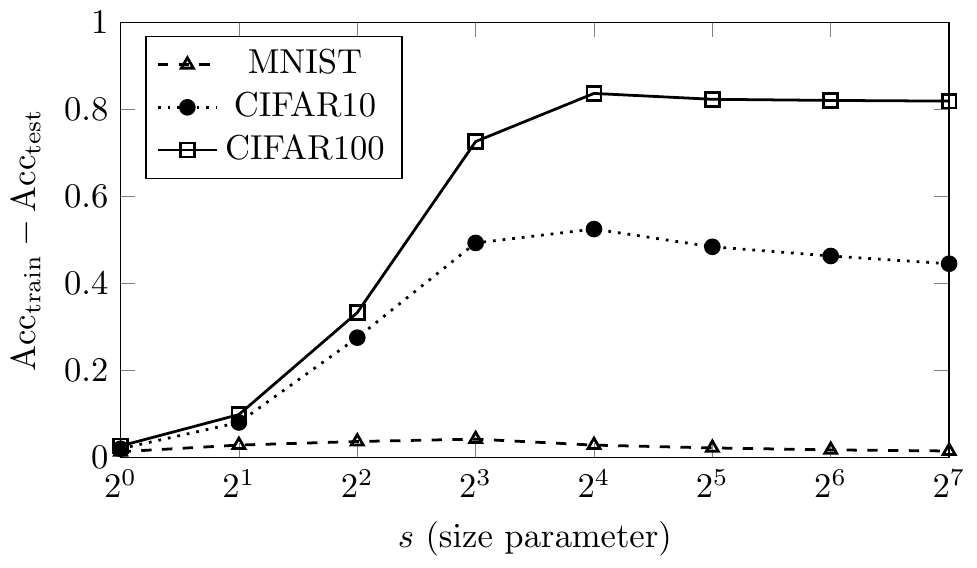}
	\centering
	\caption{Generalization error measured as the difference between training and test accuracy. On MNIST, the maximum was achieved at $s=8$ at 0.05, while for CIFAR-10 the maximum was 0.52 ($s=16$), and 0.82 ($s=16$) for CIFAR-100.}
	\label{fig:collusion-gen}
	\end{subfigure}
	\caption{Results of colluding training algorithm and membership adversary on CNNs trained on MNIST, CIFAR-10, and CIFAR-100. The size parameter was configured to take values $s=2^i$ for $i \in [0,7]$. Regardless of the models' generalization performance, when the network is sufficiently large, the attack achieves high advantage ($\ge 0.98$) without affecting predictive accuracy.}
\end{figure}

The results of our experiment are shown in Figures~\ref{fig:collusion-adv} and \ref{fig:collusion-gen}. The data shows that on all three instances, the colluding parties achieve a high membership advantage without significantly affecting model performance. The accuracy of the subverted model was only 0.014 (MNIST), 0.047 (CIFAR-10), and 0.031 (CIFAR-100) less than that of the unsubverted model. The advantage rapidly increases with the model size around $s \approx 16$ but is relatively constant elsewhere, indicating that model capacity beyond a certain point is a necessary factor in the attack.

Importantly, the results demonstrate that specific information about nearly all of the training data can be intentionally leaked through the behavior of a model that appears to generalize very well. In fact, looking at Figure~\ref{fig:collusion-gen} shows that in these instances, there is no discernible relationship between generalization error and membership advantage. The three datasets exhibit vastly different generalization behavior, with the MNIST models achieving almost no generalization error ($< 0.02$ for $s \ge 32$) and CIFAR-100 showing a large performance gap ($\ge 0.8$ for $s \ge 32$). Despite this fact, the membership adversary achieves nearly identical performance.


\section{Related Work}
\label{sect:related}


\subsection{Privacy and statistical summaries}
There is extensive prior literature on privacy attacks on statistical summaries. Komarova et al.~\cite{Komarova} looked into partial disclosure scenarios, where an adversary is given fixed statistical estimates from combined public and private sources and attempts to infer the sensitive feature of an individual referenced in those sources. A number of previous studies~\cite{wang-snp,homer08resolving,sankararaman2009genomic,ElEmam11,Gymrek321,shringarpure15} have looked into membership attacks from statistics commonly published in genome-wide association studies (GWAS). Calandrino et al.~\cite{Calandrino2011} showed that temporal changes in recommendations given by collaborative filtering methods can reveal the inputs that caused those changes. Linear reconstruction attacks~\cite{Dinur2003,Dwork2007,Smith2010} attempt to infer partial inputs to linear statistics and were later extended to non-linear statistics~\cite{linear-attacks}. While the goal of these attacks has commonalities with both membership inference and attribute inference, our results apply specifically to machine learning settings where generalization error and influence make our results relevant.

\subsection{Privacy and machine learning}
More recently, others have begun examining these attacks in the context of machine learning. Ateniese et al.~\cite{AtenieseFMSVV13} showed that the knowledge of the internal structure of Support Vector Machines and Hidden Markov Models leaks certain types of information about their training data, such as the language used in a speech dataset.

Dwork et al.~\cite{Dwork2015-2} showed that a differentially private algorithm with a suitably chosen parameter generalizes well with high probability. Subsequent work showed that similar results are true under related notions of privacy. In particular, Bassily et al.~\cite{Bassily2016} studied a notion of privacy called total variation stability and proved good generalization with respect to a bounded number of adaptively chosen low-sensitivity queries. Moreover, for data drawn from Gibbs distributions, Wang et al.~\cite{Wang2016KL} showed that on-average KL privacy is equivalent to generalization error as defined in this paper. While these results give evidence for the relationship between privacy and overfitting, we construct an attacker that directly leverages overfitting to gain advantage commensurate with the extent of the overfitting.

\subsubsection{Membership inference}
Shokri et al.~\cite{ShokriSS17} developed a membership inference attack and applied it to popular machine-learning-as-a-service APIs. Their attacks are based on ``shadow models'' that approximate the behavior of the model under attack. The shadow models are used to build another machine learning model called the ``attack model'', which is trained to distinguish points in the training data from other points based on the output they induce on the original model under attack. As we discussed in Section~\ref{sect:inclusion-eval}, our simple threshold adversary comes surprisingly close to the accuracy of their attack, especially given the differences in complexity and requisite adversarial assumptions between the attacks.

Because the attack proposed by Shokri et al.\ itself relies on machine learning to find a function that separates training and non-training points, it is not immediately clear why the attack works, but the authors hypothesize that it is related to overfitting and the ``diversity'' of the training data. They graph the generalization error against the precision of their attack and find some evidence of a relationship, but they also find that the relationship is not perfect and conclude that model structure must also be relevant. The results presented in this paper make the connection to overfitting precise in many settings, and the colluding training algorithm we give in Section~\ref{sect:collusion} demonstrates exactly how model structure can be exploited to create a membership inference vulnerability.

Li et al.~\cite{Li2013} explored membership inference, distinguishing between ``positive'' and ``negative'' membership privacy. They show how this framework defines a family of related privacy definitions that are parametrized on distributions of the adversary's prior knowledge, and they find that a number of previous definitions can be instantiated in this way.

\subsubsection{Attribute inference}
Practical model inversion attacks have been studied in the context of linear regression~\cite{fredrikson2014privacy,dptheory2016}, decision trees~\cite{mi2015}, and neural networks~\cite{mi2015}. Our results apply to these attacks when they are applied to data that matches the distributional assumptions made in our analysis. An important distinction between the way inversion attacks were considered in prior work and how we treat them here is the notion of advantage. Prior work on these attacks defined advantage as the difference between the attacker's predictive accuracy given the model and the best accuracy that could be achieved without the model. Although some prior work~\cite{fredrikson2014privacy,mi2015} empirically measured this advantage on both training and test datasets, this definition does not allow a formal characterization of how exposed the \emph{training data specifically} is to privacy risk. In Section~\ref{sect:inversion}, we define attribute advantage precisely to capture the risk to the training data by measuring the difference in the attacker's accuracy on training and test data: the advantage is zero when the attack is as powerful on the general population as on the training data and is maximized when the attack works \emph{only} on the training data.

Wu et al.~\cite{mitheory2016} formalized model inversion for a simplified class of models that consist of Boolean functions and explored the initial connections between influence and advantage. However, as in other prior work on model inversion, the type of advantage that they consider says nothing about what the model specifically leaks about its training data. Drawing on their observation that influence is relevant to privacy risk in general, we illustrate its effect on the notion of advantage defined in this paper and show how it interacts with generalization error.


\section{Conclusion and Future Directions}
\label{sect:conclusion}

We introduced new formal definitions of advantage for membership and attribute inference attacks. Using these definitions, we analyzed attacks under various assumptions on learning algorithms and model properties, and we showed that these two attacks are closely related through reductions in both directions.  Both theoretical and experimental results confirm that models become more vulnerable to both types of attacks as they overfit more. Interestingly, our analysis also shows that overfitting is not the only factor that can lead to privacy risk:  Theorem~\ref{thm:weak-stability} shows that even stable learning algorithms, which provably do not overfit, can leak precise membership information, and the results in Section~\ref{sect:inversion-attacks} demonstrate that the influence of the target attribute on a model's output plays a key role in attribute inference. 

Our formalization and analysis open interesting directions for future work. The membership attack in Theorem~\ref{thm:weak-stability} is based on a colluding pair of adversary and learning rule, \Traincol and \Advcol.  This could be implemented, for example, by a malicious ML algorithm provided by a third-party library or cloud service to subvert users' privacy. Further study of this scenario, which may best be formalized in the framework of algorithm substitution attacks~\cite{BPR14}, is warranted to determine whether malicious algorithms can produce models that are indistinguishable from normal ones and how such attacks can be mitigated.

Our results in Section~\ref{sect:dpbound} give bounds on membership advantage when certain conditions are met. These bounds apply to adversaries who may target specific individuals, bringing arbitrary background knowledge of their targets to help determine their membership status. Some types of realistic adversaries may be motivated by concerns that incentivize learning a limited set of facts about as many individuals in the training data as possible rather than obtaining unique background knowledge about specific individuals. Characterizing these ``stable adversaries'' is an interesting direction that may lead to tighter bounds on advantage or relaxed conditions on the learning rule.

\bibliographystyle{ieeetr}
\bibliography{biblio}

\newpage

\appendix

\begin{figure}
	\begin{center}
	\begin{subfigure}[b]{\subfigurewidth}
		\includegraphics[width=0.9\columnwidth]{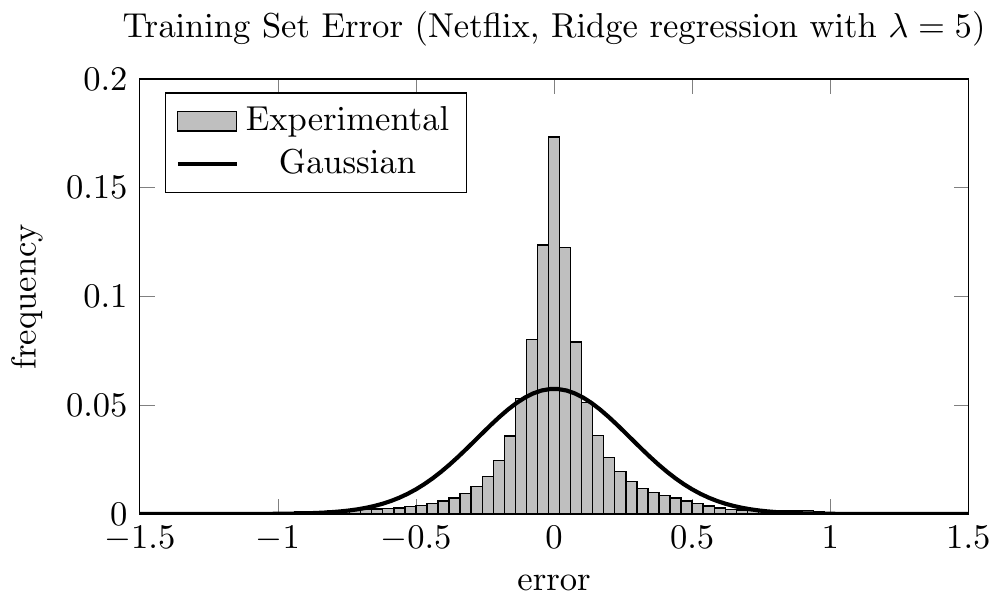}
	\end{subfigure}
	\begin{subfigure}[b]{\subfigurewidth}
		\includegraphics[width=0.9\columnwidth]{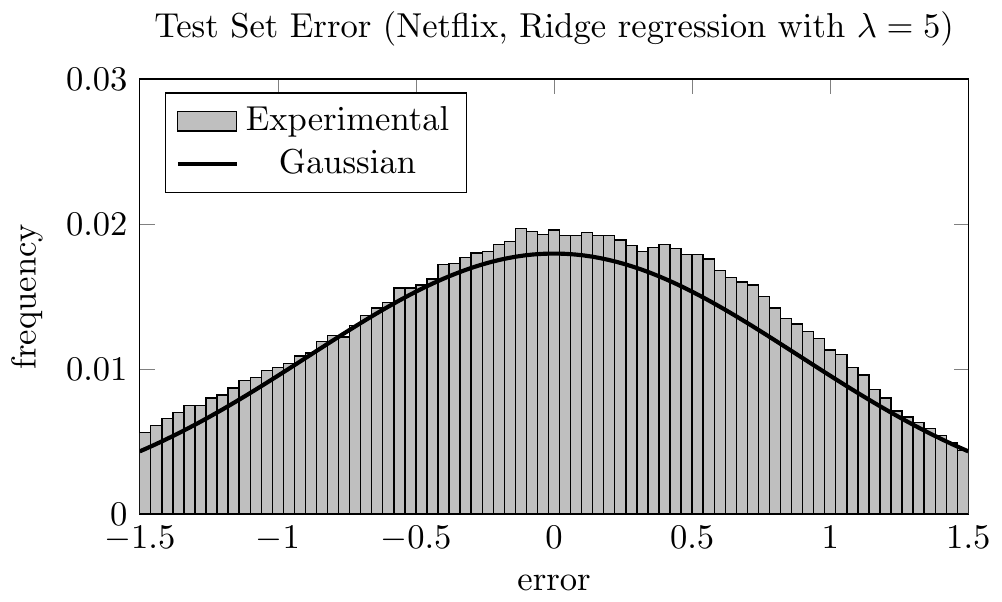}
	\end{subfigure}
	\caption{The training and test error distributions for an overfitted Ridge regression model. The histograms are juxtaposed with what we would expect if the errors were normally distributed with standard deviation $R_{emp} = 0.2774$ and $R_{cv} = 0.8884$, respectively. Note the different vertical scale for the two graphs. To minimize the effect of noise, the errors were measured using 1000 different random 75-25 splits of the data into training and test sets and then aggregated.}
	\label{fig:linreg-error-dist}
	\end{center}
	\begin{center}
	\begin{subfigure}[b]{\subfigurewidth}
		\includegraphics[width=0.9\columnwidth]{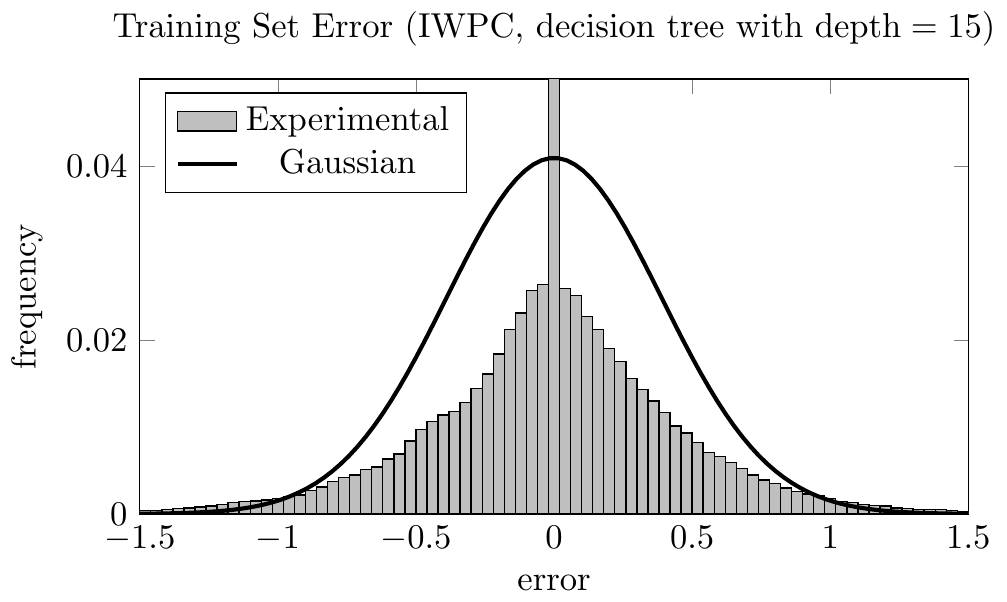}
	\end{subfigure}
	\begin{subfigure}[b]{\subfigurewidth}
		\includegraphics[width=0.9\columnwidth]{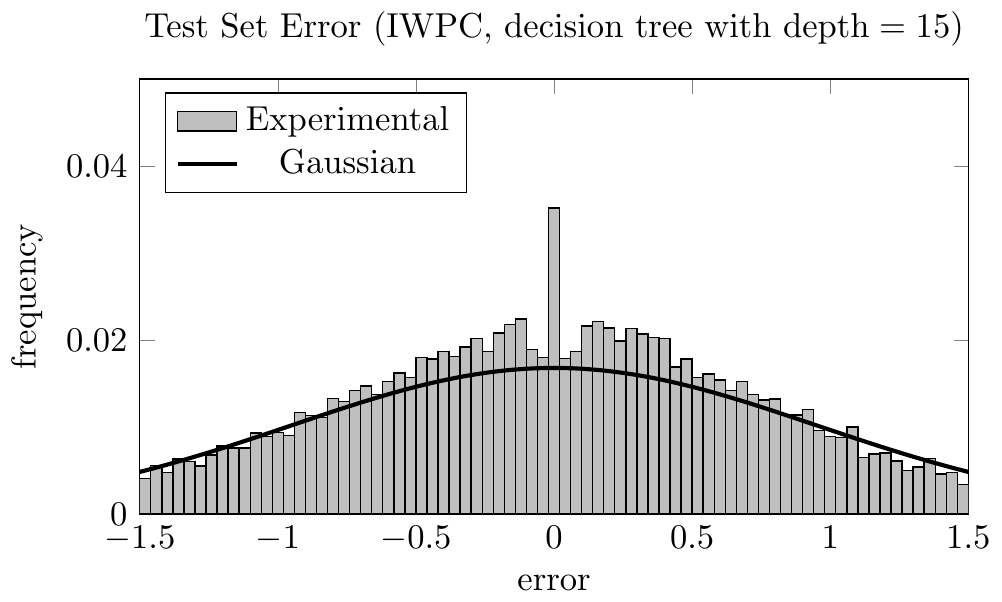}
	\end{subfigure}
	\caption{The training and test error distributions for an overfitted decision tree. The histograms are juxtaposed with what we would expect if the errors were normally distributed with standard deviation $R_{emp} = 0.3899$ and $R_{cv} = 0.9507$, respectively. The bar at $\mathrm{error} = 0$ does not fit inside the first graph; in order to fit it, the graph would have to be almost 10 times as high. To minimize the effect of noise, the errors were measured using 1000 different random 75-25 splits of the data into training and test sets and then aggregated.}
	\label{fig:tree-error-dist}
	\end{center}
\end{figure}

\end{document}